\documentclass[12pt]{article}

\usepackage[utf8]{inputenc}
\usepackage[english]{babel}  

\usepackage{natbib}

\usepackage{amsmath}
\usepackage[margin = 1.2 in]{geometry}
\usepackage{setspace}
\usepackage[shortlabels]{enumitem}
\usepackage{amssymb}
\usepackage{amsthm}

\usepackage{titling}

\newtheorem{theorem}{Theorem}

\newtheorem{definition}{Definition}

\newtheorem{lemma}{Lemma}

\usepackage[labelfont=bf]{caption}
\usepackage{pgfplots}
\usepackage{float}
\usepgfplotslibrary{fillbetween}
\usetikzlibrary{patterns}

\pgfplotsset{width = 7 cm, compat = 1.14}

\usepackage{authblk}

\usepackage{listings}
\usepackage{color} %red, green, blue, yellow, cyan, magenta, black, white
\definecolor{mygreen}{RGB}{28,172,0} % color values Red, Green, Blue
\definecolor{mylilas}{RGB}{170,55,241}

\usetikzlibrary{fadings}
\tikzfading[name=myfading, bottom color=transparent!100, top color=transparent!10]
\tikzfading[name=myalternatefading, bottom color=transparent!50, top color=transparent!10]
\tikzfading[name=mysecondalternatefading, bottom color=green!10, top color=green!10!black]
\definecolor{deepgreen}{rgb}{0,0.5,0}

\title{Simple Records Support Robust Indirect Reciprocity}

\author[1]{Daniel Clark}
\author[1]{Drew Fudenberg}
\author[1]{Alexander Wolitzky}

\affil[1]{Department of Economics, MIT, Cambridge, MA 02139, USA}

\date{\today}

\begin{document}

\maketitle

\setstretch{1.5}

\textbf{Indirect reciprocity is a foundational mechanism of human cooperation [\citenum{trivers1971evolution, Sugden1986TheWelfare, alexander1987biology, Kandori1992, Nowak1998EvolutionScoring, seabright2004company, nowak2005evolution, sigmund2010calculus, santos2018social, henrich2007humans, boyd1989evolution, bowles2013cooperative}].
Existing models of indirect reciprocity fail to robustly support social cooperation: image scoring models [\citenum{Nowak1998EvolutionScoring}]  fail to provide robust incentives, while social standing models [\citenum{Sugden1986TheWelfare, leimar2001evolution, Panchanathan2003AReciprocity, Sigmund2012MoralReciprocity}] are not informationally robust. Here we provide a new model of indirect reciprocity based on simple, decentralized \emph{records}: each individual's record depends on their own past behavior alone, and not on their partners' past  behavior or their partners' partners' past behavior. When social dilemmas exhibit a coordination motive (or \emph{strategic complementarity}), tolerant trigger strategies based on simple records can robustly support positive social cooperation and exhibit strong stability properties. In the opposite case of \emph{strategic substitutability}, positive social cooperation cannot be robustly supported. Thus, the strength of short-run coordination motives in social dilemmas determines the prospects for robust long-run cooperation.}

People (and perhaps also other animals) often trust each other to cooperate even when they know they will never meet again.  Such indirect reciprocity relies on individuals having some information about how their partners have behaved in the past. Existing models of indirect reciprocity fall into two paradigms. In the image scoring paradigm, each individual carries an \emph{image} that improves when they help others, and (at least some) individuals help only those with good images. In the standing paradigm, each individual carries a \emph{standing} that (typically) improves when they help others with good standing, but not when they help those with bad standing, and individuals with good standing help only other good-standing individuals.

Neither of these paradigms provides a robust explanation for social cooperation. In image-scoring models, there is no reason for an individual to only help  partners with good images: since the partner's image does not affect one's future payoff, helping some partners and not others is optimal only if one is completely indifferent between helping and not helping. In game-theoretic terms, individuals never have strict incentives to follow image-scoring strategies, and hence such strategies can form at best a weak equilibrium.
Closely related to this point, image-scoring equilibria are unstable in several environments [\citenum{leimar2001evolution, Panchanathan2003AReciprocity}].
Standing models do yield strict, stable equilibria, but they fail to be informationally robust: computing an individual's standing requires knowledge of not only their own past behavior, but also their past partners' behavior, their partners' partners' behavior, and so on ad infinitum. Such infinite-order information is likely unavailable in many societies [\citenum{nowak2005evolution}].

We develop a new theoretical paradigm for modeling indirect reciprocity that supports positive social cooperation as a strict, stable equilibrium while relying only on simple, \emph{individualistic} information: when two players meet, they observe each other's records and nothing else, and each individual's record depends only on their own past behavior. (Individualistic information is also called ``first-order'' [\citenum{Takahashi2010, Heller2018, bhaskar2018, ohtsuki2006leading}].)

As our model of individual interaction, we use the classic prisoner's dilemma (``PD'') with actions $C,D$ (``Cooperate,'' ``Defect'') and a standard payoff normalization, with $g,l > 0$ and $g - l < 1$---see the top panel in \textbf{Fig. \ref{Prisoner's Dilemma Figure}}. This canonical game can capture many two-sided interactions, such as business partnerships 
 [\citenum{klein1981role}], management of public resources 
 [\citenum{hardin1968tragedy, ostrom1990governing}], and risk-sharing in developing societies [\citenum{coate1993reciprocity}],
 as well as many well-documented animal behaviors [\citenum{dugatkin1997cooperation}]. 
 
\definecolor{darkpastelgreen}{rgb}{0.01, 0.75, 0.24}

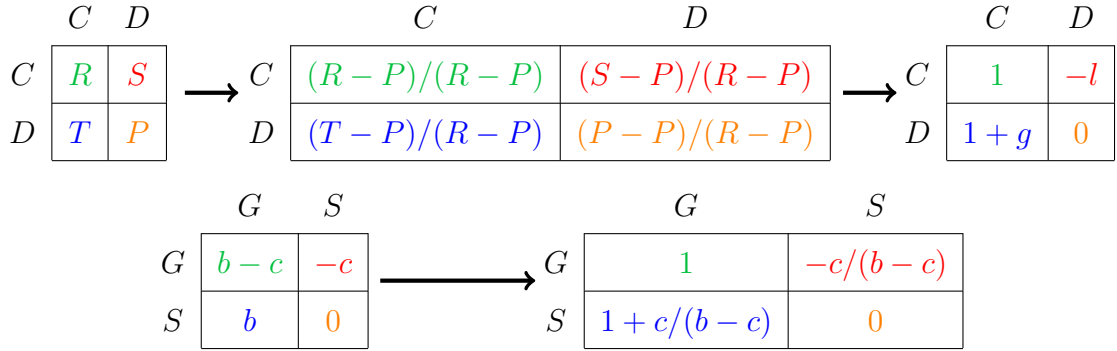
\begin{figure}[H]
\centering
\begin{tikzpicture}
\node(a) at (0,1.5){
\begin{tikzpicture}
\node (a) at (-9,0){
\begin{tabular}{ccc}
& $C$ & $D$ \\ \cline{2-3}
$C$ & \multicolumn{1}{|c}{\textcolor{darkpastelgreen}{$R$}} & \multicolumn{1}{|c|}{\textcolor{red}{$S$}} \\ 
\cline{2-3}
$D$ & \multicolumn{1}{|c}{\textcolor{blue}{$T$}} & \multicolumn{1}{|c|}{\textcolor{orange}{$P$}} \\ 
\cline{2-3}
\end{tabular}
};
\node(b) at (-3,0){
\begin{tabular}{ccc}
& $C$ & $D$ \\ \cline{2-3}
$C$ & \multicolumn{1}{|c}{\textcolor{darkpastelgreen}{$(R-P)/(R-P)$}} & \multicolumn{1}{|c|}{\textcolor{red}{$(S-P)/(R-P)$}} \\ 
\cline{2-3}
$D$ & \multicolumn{1}{|c}{\textcolor{blue}{$(T-P)/(R-P)$}} & \multicolumn{1}{|c|}{\textcolor{orange}{$(P-P)/(R-P)$}} \\ 
\cline{2-3}
\end{tabular}
};
\node(c) at (3.25,0){
\begin{tabular}{ccc}
& $C$ & $D$ \\ \cline{2-3}
$C$ & \multicolumn{1}{|c}{\textcolor{darkpastelgreen}{$1$}} & \multicolumn{1}{|c|}{\textcolor{red}{$-l$}} \\ 
\cline{2-3}
$D$ & \multicolumn{1}{|c}{\textcolor{blue}{$1+g$}} & \multicolumn{1}{|c|}{\textcolor{orange}{$0$}} \\ 
\cline{2-3}
\end{tabular}
};
\node(d) at (-7.75,-.25){};
\node(e) at (-6.75,-.25){};
\node(f) at (1,-.25){};
\node(g) at (2,-.25){};
\draw[->,ultra thick](d)--(e);
\draw[->,ultra thick](f)--(g);
\end{tikzpicture}
};
\node(b) at (0,-1){
\begin{tikzpicture}
\node (a) at (-4,0){
\begin{tabular}{ccc}
& $G$ & $S$ \\ \cline{2-3}
$G$ & \multicolumn{1}{|c}{\textcolor{darkpastelgreen}{$b-c$}} & \multicolumn{1}{|c|}{\textcolor{red}{$-c$}} \\ 
\cline{2-3}
$S$ & \multicolumn{1}{|c}{\textcolor{blue}{$b$}} & \multicolumn{1}{|c|}{\textcolor{orange}{$0$}} \\ 
\cline{2-3}
\end{tabular}
};
\node (b) at (2.5,0) 
{
\begin{tabular}{ccc}
& $G$ & $S$ \\ \cline{2-3}
$G$ & \multicolumn{1}{|c}{\textcolor{darkpastelgreen}{$1$}} & \multicolumn{1}{|c|}{\textcolor{red}{$-c/(b - c)$}} \\ 
\cline{2-3}
$S$ & \multicolumn{1}{|c}{\textcolor{blue}{$1+c/(b-c)$}} & \multicolumn{1}{|c|}{\textcolor{orange}{$0$}} \\ 
\cline{2-3}
\end{tabular}
};
\node(c) at (-2.5,-.25){};
\node(d) at (-.15,-.25){};
\draw[->,ultra thick](c)--(d);
\end{tikzpicture}
};
\end{tikzpicture}
\caption{\textbf{The prisoner's dilemma.} The top panel shows how any prisoner's dilemma can be represented by the standard normalization with $g = (T - R) / (R - P)$ and $l = (P - S) / (R - P)$, where $T>R>P>S$. The bottom panel illustrates this normalization for ``donation games'' in which choosing $G$ ($Give$) instead of $S$ ($Shirk$) incurs a personal cost $c$ and gives benefit $b>c$ to the opponent.}
\label{Prisoner's Dilemma Figure}
\end{figure}

A critical feature of the PD is whether it exhibits \emph{strategic complementarity} or \emph{strategic substitutability}. Strategic complementarity means that the gain from playing $D$ is greater when the opponent also plays $D$. In the PD payoff matrix displayed in \textbf{Fig. \ref{Prisoner's Dilemma Figure}}, this corresponds to the condition \begin{equation} g<l. \tag{\textbf{Strategic Complementarity}} \end{equation} Strategic complementarity is a common case in realistic social dilemmas: it implies that although $D$ is selfishly optimal regardless of the partner's action  (a defining feature of the PD), the social dilemma nonetheless retains some aspect of a coordination or stag-hunt game, so that playing $C$ is less costly when one's partner also plays $C$. For example, mobbing a predator is always risky (hence costly) for each individual, but it is much less risky when others also mob [\citenum{zahavi1995altruism}]. The opposite case of \emph{strategic substitutability} arises when the gain from playing $D$ is greater when the opponent plays $C$: mathematically, this occurs when \[g>l.\tag{\textbf{Strategic Substitutability}} \] The distinction between strategic complementarity and substitutability has long been known to be of critical importance in economics [\citenum{bulow1985multimarket, fudenberg1984fat}], 
but its implications for the evolution of cooperation have not previously been assessed.

In our model, each player's record is an integer that tracks how many times that player has defected. Newborn players have record $0$. Whenever an individual plays $D$, their record increases by $1$. Whenever an individual plays $C$, their record remains constant with probability $1-\varepsilon$ and increases by $1$ with probability $\varepsilon$; thus, $\varepsilon \in (0,1)$ measures the amount of noise in the system, which can reflect either errors in recording or errors in executing the intended action. An individual's record is considered to be ``good'' if the number of times the individual has been recorded as playing $D$ is less than some pre-specified threshold $K$: this individualistic scoring is similar to image-scoring models. When two individuals meet, they both play $C$ if and only if they both have good records: this conditioning on the partner's record is similar to standing models.
 We call these strategies \emph{tolerant trigger strategies} or $GrimK$, as they are a form of the well-known grim trigger strategies [\citenum{friedmand1971supergames}] with a ``tolerance'' of $K$ recorded plays of $D$.

We analyze the steady-state equilibria of a system where the total population size is constant, but each individual has a geometrically distributed lifetime [\citenum{ohtsuki2006leading}]. To ensure robustness, we insist that equilibrium behavior is strictly optimal at every record; in classical (normal-form) games, this implies that the equilibrium is evolutionarily stable [\citenum{smith1982evolution, weibull1997evolutionary}].

We show that $GrimK$ strategies can form a strict steady-state equilibrium if and only if the PD exhibits substantial \emph{strategic complementarity}, in that the gain from playing $D$ rather than $C$ is significantly greater when the opponent plays $D$: the precise condition required in the PD payoff matrix displayed in \textbf{Fig. \ref{Prisoner's Dilemma Figure}} is \[g<\frac{l}{1+l}.\] 

Most previous studies of indirect reciprocity restrict attention to the ``donation game'' instance of the PD [\citenum{sigmund2010calculus}] where $g=l$---see the bottom panel in \textbf{Fig. \ref{Prisoner's Dilemma Figure}}. Our analysis reveals this to be a knife-edge case that obscures the distinction between strategic complementarity ($g<l$) and substitutability ($g>l$). (However, the $g\neq l$ case has also received significant attention: for example, the seminal article of Axelrod and Hamilton [\citenum{Axelrod1981TheCooperation}] took $g=1$ and $l=1/2$.) We show that the tolerance level $K$ can be tuned so that $GrimK$ strategies robustly support positive social cooperation in the presence of sufficiently strong strategic complementarity. 

To see how to tune the threshold $K$, note that since even  individuals who always try to cooperate are sometimes recorded as playing $D$ due to noise, $K$ must be large enough that the steady-state share of the population with good records is sufficiently high: with any fixed value of $K$, a population of sufficiently long-lived players would almost all have bad records. However, $K$ also cannot be too high, as otherwise an individual with a very good record (that is, a very low number of $D$'s) can safely play $D$ until their record approaches the threshold. Another constraint is that an individual with record $K-1$ who meets a partner with a bad record must not be tempted to deviate to $C$ to preserve their own good record. These constraints lead to an upper bound on the maximum share of cooperators in equilibrium. As lifetimes become long and noise becomes small, this upper bound converges to $0$ whenever $g > l / (1 + l)$ and to $l / (1 + l)$ whenever $g < l/(1 + l)$---see Table \ref{Upper Bounds on Cooperation Table}---and we show that this share of cooperators can in fact be attained in equilibrium in the $(\gamma, \varepsilon) \rightarrow (1,0)$ limit. Thus, greater strategic complementarity (higher $l$ and lower $g$) not only helps support some cooperation; it also increases the maximum level of cooperation in the limit, as shown in \textbf{Fig. \ref{Limit Performance Figure}}.

\begin{table}[H]
\centering
\begin{tikzpicture}
\node(a) at (-1.5,-.25){
\includegraphics[scale=1]{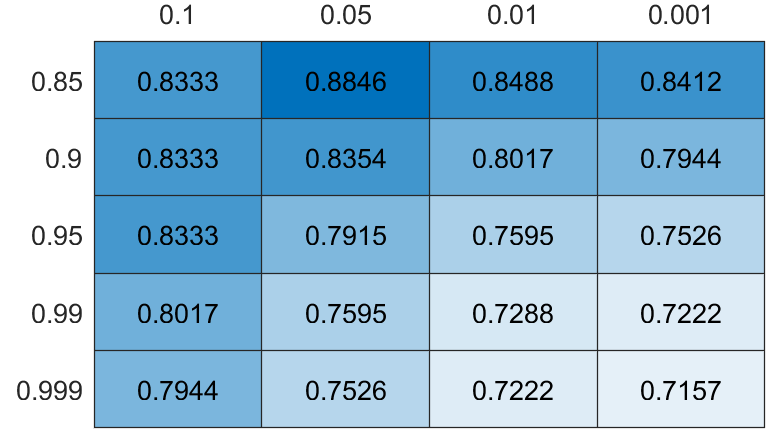}
};
\node(b) at (6,-.25){
\includegraphics[scale=1]{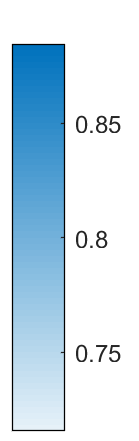}
};
\node[rotate=90](c) at (-7,-.35){Survival Probability ($\gamma$)};
\node(d) at (-1,3.25){Noise ($\varepsilon$)};
\node[rotate=90](c) at (5,-.35){Level of Cooperation};
\end{tikzpicture}
\caption{\textbf{Upper bounds on cooperation.} The entries are upper bounds on the share of cooperators possible in equilibria for various $\gamma$ and $\varepsilon$ values when $g = 0.5$ and $l = 2.5$, with a darker shade indicating a higher value as shown in the scale at right. As we move to the bottom right, the upper bound converges to $l/(1+l) \approx .7143$, which is the maximum share of cooperators sustainable in the limit,  but away from the limit the upper bound can be different (the values in this table are all higher, but this is not the case for small $\gamma$ or large $\varepsilon$).}
\label{Upper Bounds on Cooperation Table}
\end{table}

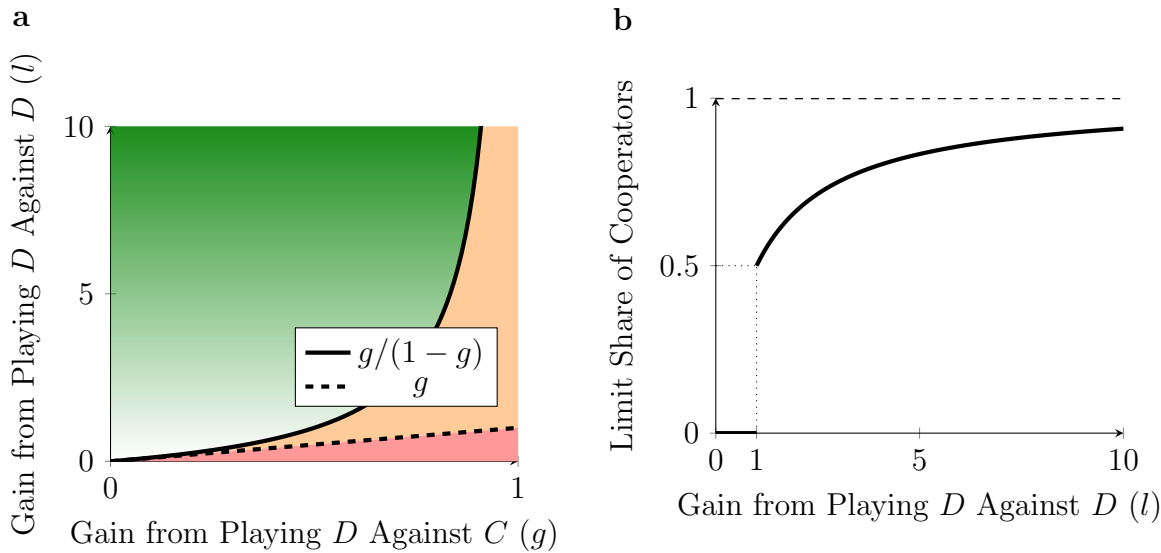
\begin{figure}[H]
\centering
\begin{tikzpicture}
\node(a) at (-4.5,0){
\begin{tikzpicture}
\begin{axis}[
    axis lines = left,
    xlabel = {Gain from Playing $D$ Against $C$ ($g$)},
    ylabel = {Gain from Playing $D$ Against $D$ ($l$)},
    xlabel style={align=center},
    ylabel style={align=center},
    xtick = {0,1},
    ytick = {0,5,10},
    xmin=0,xmax=1,
    ymin=0,ymax=10,
    legend style={at={(0.95,0.4)}},
]
\addplot[name path=C1,mark=none,solid,ultra thick] [
    domain=0:11/12, 
    samples=100, 
]
{x/(1-x)};
\addlegendentry{$g/(1-g)$}

\addplot[name path=C2,mark=none,dashed,ultra thick] [
    domain=0:1, 
    samples=100, 
]
{x};
\addlegendentry{$g$}

\addplot[name path=C3,mark=none,solid,transparent] [
    domain=0:1, 
    samples=100, 
]
{0};

\addplot[name path=C4,mark=none,solid,transparent] [
    domain=0:1, 
    samples=100, 
]
{10};

\addplot[path fading = myfading, fill=deepgreen]fill between[of=C1 and C4, soft clip={domain=0:10/11}];

\addplot[fill=orange!40]fill between[of=C2 and C1, soft clip={domain=0:12/13}];

\addplot[fill=orange!40]fill between[of=C2 and C4, soft clip={domain=10/11:1}];

\addplot[fill=red!40]fill between[of=C3 and C2, soft clip={domain=0:1}];

\end{axis}
\end{tikzpicture}
};
\node(c) at (3.5,0){
\begin{tikzpicture}
\begin{axis}[
    axis lines = left,
    xlabel = {Gain from Playing $D$ Against $D$ ($l$)},
    ylabel = {Limit Share of Cooperators},
    xlabel style={align=center},
    ylabel style={align=center},
    xtick = {0,1,5,10},
    ytick = {0,.5,1},
    legend style={at={(2.0,0.65)}},
]

\addplot[name path=C1,mark=none,solid,ultra thick] [
    domain=1:10, 
    samples=100, 
]
{x/(1+x)};

\addplot[name path=C2,mark=none,solid,ultra thick] [
    domain=0:1, 
    samples=100, 
]
{0};

\addplot[name path=C4,mark=none,dashed,thick] [
    domain=0:10, 
    samples=100, 
]
{1};

\addplot[name path=C5,mark=none,dotted,thin,forget plot] [
    domain=0:1, 
    samples=100, 
    ]
    {1/2};
    
\draw [dotted,thin] (1,0) -- (1,.5);

\end{axis}
\end{tikzpicture}
};
\node(c) at (-8,3.75){\textbf{a}};
\node(d) at (0,3.75){\textbf{b}};
\end{tikzpicture}
\caption{\textbf{Limit performance of $GrimK$ strategies.} \textbf{a}, In the green region ($l > g / (1 - g)$), $GrimK$ strategies sustain a positive limit share of cooperators, which increases with $l$, as indicated by a deeper shade of green. In the orange region ($g < l < g / (1 - g)$), the limit share of cooperators with $GrimK$ is $0$, but other strategies may sustain positive cooperation in the limit. In the red region ($l \leq g$), individualistic records preclude cooperation. \textbf{b}, The limit share of cooperators as a function of $l$ when $g = 1/2$.  At $l = 1$, there is a discontinuity; as $l \rightarrow \infty$, the limit share of cooperators approaches $1$.}
\label{Limit Performance Figure}
\end{figure}

Note that with image-scoring strategies an individual's image improves when they cooperate, in contrast to $GrimK$ strategies where cooperation only slows the deterioration of one's image. Modifying $GrimK$ strategies by specifying that cooperation improves an individual's image does not help support cooperation: our results for maximum cooperation under $GrimK$ strategies also hold for this more complicated class of strategies.

$GrimK$ strategies also satisfy desirable stability and convergence properties. These derive from a key monotonicity property of $GrimK$ strategies: when the distribution of individual records is more favorable 
today, the same will be true tomorrow, because players with better records both behave more cooperatively and induce more cooperative behavior from their partners. (See \textbf{Methods} for a precise statement.) From this observation it can be shown that, whenever the initial distribution of records is more favorable than the best steady-state record distribution, the record distribution converges to the best steady state. Similarly, whenever the initial distribution is less favorable than the worst steady-state, convergence to the worst steady state obtains. See \textbf{Fig. \ref{Convergence Figure}}. These additional robustness properties are not shared by more complicated, non-monotone strategies that can sometimes support  cooperation for a wider range of parameters than $GrimK$.

\begin{figure}[H]
\centering
\begin{tikzpicture}
\node(a) at (-4,0){
\includegraphics[width=.5\textwidth]{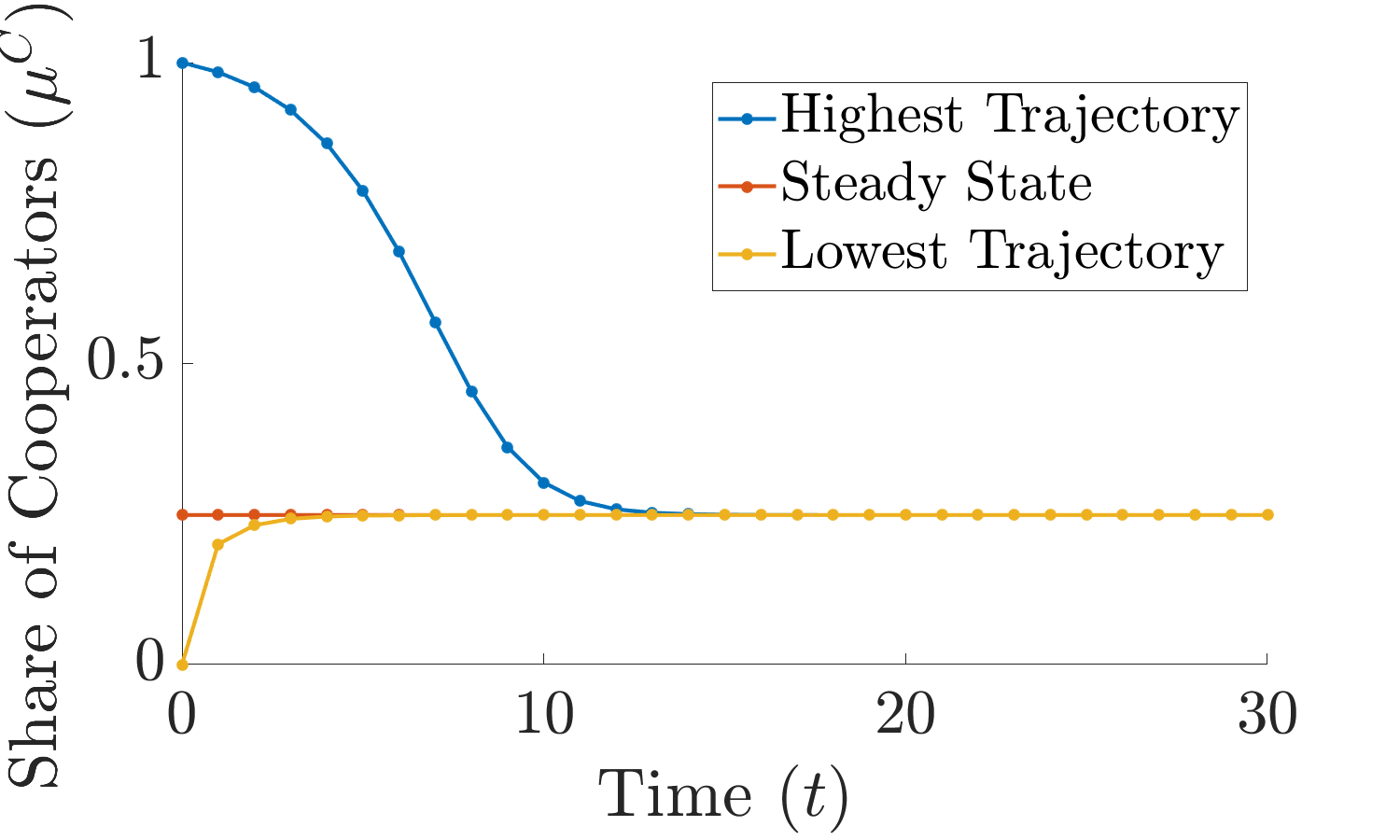}
};
\node(b) at (4,0){
\includegraphics[width=.5\textwidth]{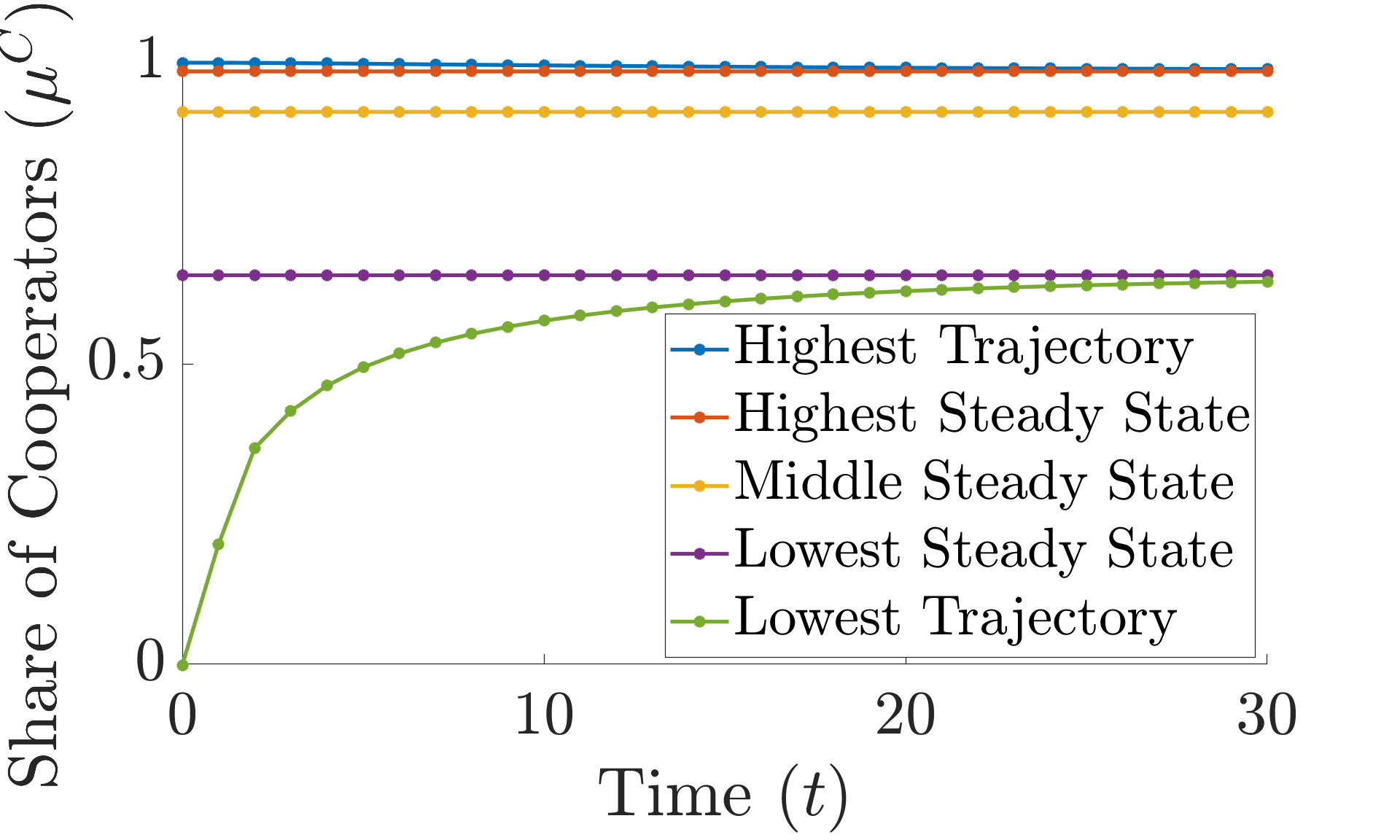}
};
\node(c) at (-8,3){\textbf{a}};
\node(d) at (0,3){\textbf{b}};
\end{tikzpicture}
\caption{\textbf{Convergence of the share of cooperators.} \textbf{a} depicts trajectories for the share of cooperators when $\gamma = 0.8$, $\varepsilon = 0.02$, and players use the $Grim1$ strategy; \textbf{b} does the same for the $Grim2$ strategy. \textbf{a}, All trajectories converge to the unique steady state; \textbf{b}, there are three steady states. Here ``high'' trajectories converge to the most cooperative steady state, while ``low'' trajectories converge to the least cooperative steady state.  See \textbf{Methods} for details.
}
\label{Convergence Figure}
\end{figure}

We also analyze evolutionary properties of $GrimK$ equilibria. When $g < l / (1 + l)$, there is a sequence of $GrimK$ equilibria that are ``steady-state robust to mutants'' and attains the maximum limit cooperation share of $l / (1 + l)$.  By this we mean that, when a small fraction of players adopt some mutant $GrimK'$ strategy where $K' \neq K$, there is a steady-state distribution of records where it remains strictly optimal to play according to $GrimK$. We also perform simulations of dynamic evolution when a  population  playing a $GrimK$ equilibrium is infected by a mutant population playing $GrimK'$ for some $K'\neq K$. (See \textbf{Supplementary Information} and \textbf{Extended Data Fig. \ref{Dynamics Figure}}.)

Although our main analysis takes the basic unit of social interaction to be the standard 2-player PD, many social interactions involve multiple players: the management of the commons and other public resources is a leading example [\citenum{hardin1968tragedy, ostrom1990governing}]. In the \textbf{Supplementary Information} we establish that, when strategic complementarity is sufficiently strong, robust cooperation in the multiplayer public goods game can be supported by a simple variant of $GrimK$ strategies, wherein a player contributes to the public good if and only if all of their current partners have good records. In contrast, with strategic substitutability the unique strict equilibrium involves zero contribution. As the $n$-player public good game is a generalization of the PD, this implies that individualistic records preclude cooperation in the PD with strategic substitutability, as indicated in the red region in \textbf{Fig. \ref{Limit Performance Figure}a}.

We have shown how individualistic records robustly support indirect reciprocity in supermodular PD and multiplayer public goods games. To place our results in context, recall that scoring models do not provide robust incentives, while standing models compute records as a recursive function of a player's partners' past actions and standing, their partners' actions and standing, and so on, and thus require more information than may typically be available. The simplicity and power of individualistic records suggests that they may be usefully adapted to specific settings where cooperation is based on indirect reciprocity,  such as online rating systems [\citenum{resnick2000reputation, dellarocas2005reputation}], credit ratings [\citenum{bhaskar2018, klein1992promise}], decentralized currencies [\citenum{kocherlakota1998incomplete, biais2019blockchain}], and monitoring systems for conflict resolution [\citenum{fearon1996explaining}]. Individualistic records may also prove useful in modeling the role of costly punishment in the evolution of cooperation [\citenum{fehr2000fairness, bhui2019exploitation, boyd2003evolution, henrich2006costly}].

\section*{References}

\setcitestyle{numbers}
\bibliographystyle{plainnat}
\bibliography{Simple_Records_Support_Robust_Indirect_Reciprocity}
% \printbibliography

\noindent
\textbf{Acknowledgements.}
\newline
\noindent
This work was supported by  National Science Foundation grants SES-1643517 and SES-1555071 and Sloan Foundation grant 2017-9633.
\newline

\noindent
\textbf{Author contributions.}
\newline
\noindent
All authors contributed equally to the work presented in this paper.
\newline

\noindent
\textbf{Author information.}
\newline
\noindent
The authors declare no competing financial interests. Correspondence and requests for materials should be addressed to D.F. (drew.fudenberg@gmail.com).

\setstretch{1.5}

\section*{Methods}

Here we summarize the model and mathematical results; further details are provided in the \textbf{Supplementary Information}.

\subsection*{A Model of Social Cooperation with Individualistic Records}

Time is discrete and doubly infinite: $t\in \{\ldots,-2, -1,0,1,2,\ldots\}$. There is a population of individuals of unit mass, each with survival probability $\gamma \in (0,1)$, so each individual's lifespan is  geometrically distributed with mean $1/(1-\gamma)$. An inflow of $1-\gamma$ newborn players each period keeps the total population size constant. We thus have an infinite-horizon dynamic model with overlapping generations of players [\citenum{Kandori1992, ohtsuki2004should, ohtsuki2006leading, brandt2005indirect, fudenberg2018learning}].

Every period, individuals randomly match in pairs to play the PD (\textbf{Fig. \ref{Prisoner's Dilemma Figure}}). Each individual carries a \emph{record} $k\in \mathbb{N} := \{0,1,2,...\}$. Newborns have record $0$. Whenever an individual plays $D$, their record increases by $1$. Whenever an individual plays $C$, their record remains constant with probability $1-\varepsilon$ and increases by $1$ with probability $\varepsilon$; thus, $\varepsilon \in (0,1)$ measures the amount of noise in the system [\citenum{le2007evolutionary, fudenberg1990evolution, Fudenberg2012SlowWorld, mcnamara2004variation, bendor1991doubt}]. The assumption that only plays of $C$ are hit by noise simplifies some formulas but does not affect any of our results.

When two players meet, they observe each other's records and nothing else. A \emph{strategy} is a mapping $\mathbf{s}: \mathbb{N} \times \mathbb{N} \to \{C,D\}$, with the convention that the first component of the domain is a player's own record and the second component is the current opponent's record. We assume that all players use the same strategy, noting that this must be the case in every strict equilibrium in a symmetric, continuum-agent model like ours. (Of course, players who have different records and/or meet opponents with different records may take different actions.) 

The \emph{state} of the system $\mu \in \Delta(\mathbb{N})$ describes the share of the population with each record, where $\mu_{k}\in [0,1]$ denotes the share with record $k$. When all players use strategy $\mathbf{s}$, let $f_{\mathbf{s}}: \Delta(\mathbb{N}) \to \Delta(\mathbb{N})$ denote the resulting \emph{update map} governing the evolution of the state. (The formula for $f_{\mathbf{s}}(\mu)$ is in the \textbf{Supplementary Information}.) A \emph{steady state} under strategy $\mathbf{s}$ is a state $\mu$ such that $f_{\mathbf{s}}(\mu)=\mu$.

Given a strategy $\mathbf{s}$ and state $\mu$, the expected flow payoff of a player with record $k$ is $\pi_{k}(\mathbf{s},\mu)=\sum_{k'}\mu_{k'}u(\mathbf{s}(k,k'),\mathbf{s}(k',k))$, where $u$ is the PD payoff function. Denote the probability that a player with current record $k$ has record $k'$ $t$ periods in the future by $\phi_{k}(\mathbf{s},\mu)^{t}(k')$. The continuation payoff of a player with record $k$ is then $V_{k}(\mathbf{s},\mu) = (1-\gamma)\sum_{t=0}^{\infty}\gamma^{t}\sum_{k'}\phi_{k}(\mathbf{s},\mu)^{t}(k')\pi_{k'}(\mathbf{s},\mu).$
Note that we have normalized continuation payoffs by $(1-\gamma)$ to express them in per-period terms. A player's objective is to maximize their expected lifetime payoff.

A pair $(\mathbf{s},\mu)$ is an \textit{equilibrium} if $\mu$ is a steady-state under $\mathbf{s}$ and, for each own record $k$ and opponent's record $k'$, the prescribed action $\mathbf{s}(k,k') \in \{C,D\}$ maximizes the expected lifetime payoff from the current period onward, given by $(1-\gamma)u(a,\mathbf{s}(k',k)) + \gamma \sum_{k''}\left( \rho(k,a)[k'']\right) V_{k''}(\mathbf{s},\mu)$, over $a\in\{C,D\}$, where $\rho(k,a)[k'']$ denotes the probability that a player with record $k$ who takes action $a$ acquires next-period record $k''$. Note that this expression depends on the opponent's record only through the predicted current-period opponent action, $\mathbf{s}(k',k)$. In addition, the ratio $(1-\gamma)/\gamma $ captures the weight that players place on their current payoff relative to their  continuation payoff from tomorrow on. We study limits where this ratio converges to $0$, as opposed to time-average payoffs which give exactly $0$ weight to any one period's payoff, because in the latter case optimization and equilibrium impose unduly weak restrictions [\citenum{fudenberg1986folk}]. 
An equilibrium is \emph{strict} if the maximizer is unique for all pairs $(k,k')$, i.e. the optimal action is always unique. Note that this equilibrium definition allows agents to maximize over all possible strategies, as opposed to only strategies from some pre-selected set. We focus on strict equilibria because they are robust, and they remain equilibria under ``small" perturbations of the model. Note that the strategy $\emph{Always Defect}$, i.e. $s(k,k')=D$ for all $(k,k')$, together with any steady state is always a strict equilibrium. Lemma \ref{Equilibrium Characterization} in the \textbf{Supplementary Information} characterizes the steady states for any $GrimK$ strategy, as well as the $\gamma,\varepsilon,g,l$ parameters for which the steady states are equilibria.

\subsection*{Limit Cooperation under $GrimK$ Strategies}

Under $GrimK$ strategies, a matched pair of players cooperate if and only if both records are below a pre-specified cutoff $K$: that is, $s(k,k')=C$ if $\max\{k,k'\}<K$, and $s(k,k')=D$ if $\max\{k,k'\}\geq K$.  

We call an individual a \emph{cooperator} if their record is below $K$ and a \emph{defector} otherwise. Note that each individual may be a cooperator for some periods of their life and a defector for other periods, rather than being pre-programmed to cooperate or defect for their entire life.

Given an equilibrium strategy $GrimK$, let $\mu^{C}=\sum_{k=0}^{K-1}\mu_{k}$ denote the corresponding steady-state share of cooperators. Note that, in a steady state with cooperator share $\mu^{C}$, mutual cooperation is played in share $(\mu^{C})^2$ of all matches. Let $\overline{\mu}^{C}(\gamma, \varepsilon)$ be the maximal share of cooperators in any tolerant trigger equilibrium (allowing for every possible $K$) when the survival probability is $\gamma$ and the noise level is $\varepsilon$.

Theorem \ref{Limit Performance Result} in the \textbf{Supplementary Information}  characterizes the performance of equilibria in $GrimK$ strategies in the double limit where the survival probability approaches $1$---so that players expect to live a long time and the ``shadow of the future" looms large---and the noise level approaches $0$---so that records are reliable enough to form the basis for incentives. (This long-lifespan/low-noise limit is the leading case of interest in theoretical analyses of indirect reciprocity [\citenum{ohtsuki2004should, ohtsuki2006leading, Ellison1994, horner2006folk, Takahashi2010}].)
The theorem  shows that, in the double limit $(\gamma,\varepsilon)\to (1,0)$, $\bar{\mu}^{C}(\gamma,\varepsilon)$ converges to $l/(1+l)$ when $g<l/(1+l)$, and converges to $0$ when $g>l/(1+l)$. The formal statement and proof of this result are contained in the \textbf{Supplementary Information}.

Barring knife-edge cases, tolerant trigger strategies can thus robustly support positive cooperation in the double limit $(\gamma,\varepsilon)\to (1,0)$ if and only if the gain from defecting against a partner who cooperates is significantly smaller than the loss from cooperating against a partner who defects: $g<l/(1+l)$. Moreover, the maximum level of cooperation in this case is $l/(1+l)$. Here we explain the logic of this result.

We first show that $g < \mu^C$ in any $GrimK$ equilibrium. Newborn individuals have continuation payoff equal to the average payoff in the population, which is $\left( \mu^{C} \right)^2$. Thus, since a newborn player plays $C$ if and only if matched with a cooperator, $\left( \mu^{C} \right)^{2} = (1 - \gamma) \mu^{C} + \gamma \mu^{C} V_{0}^{C} + \gamma (1 - \mu^{C}) V_{0}^{D}$, where $V_{0}^{C}$ and $V_{0}^{D}$ are the expected continuation payoffs of a newborn player after playing $C$ and $D$, respectively. Newborn players have the highest continuation payoff in the population, so $V_{0}^{C} \leq V_{0} = \left( \mu^{C} \right)^{2}$. For a newborn player to prefer not to cheat a cooperative partner, it must be that $V_{0}^{D} < V_{0}^{C} - (1 - \gamma) g / \gamma$, so when $\mu^{C} < 1$ (as is necessarily the case with any noise),
\begin{align*}
     \left( \mu^C \right)^{2} & < (1-\gamma) \mu^{C} + \gamma \left( \mu^{C} \right)^{2} - (1 - \gamma) (1 - \mu^{C})g .
\end{align*}
This inequality can hold only if $g < \mu^C$.

We next show that $\gamma (1 - \varepsilon) \mu^C < l/(1+l)$ in any $GrimK$ equilibrium. The continuation payoff $V_{K-1}$ of an individual with record $K-1$ satisfies $V_{K-1}=(1-\gamma)\mu^C + \gamma (1 - \varepsilon) \mu^C V_{K-1}$, or $V_{K-1}=(1-\gamma)\mu^C/(1-\gamma (1 - \varepsilon) \mu^C)$. A necessary condition for an individual with record $K-1$ to prefer to play $D$ against a defector partner is $(1-\gamma)(-l)+\gamma (1 - \varepsilon) V_{K-1} <0$, or $l>\gamma (1 - \varepsilon) V_{K-1}/(1-\gamma)$. Combining this inequality with the expression for $V_{K-1}$ yields $\gamma (1 - \varepsilon) \mu^C < l/(1+l)$, which in the $(\gamma, \varepsilon) \rightarrow (1,0)$ limit gives $\mu^{C} \leq l / (1 + l)$.

    We have established that tolerant trigger strategies can support positive cooperation in the $(\gamma, \varepsilon) \rightarrow (1,0)$ limit only if $g\leq l/(1+l)$, and that the maximum cooperation share cannot exceed $l/(1+l)$. The proof of Theorem \ref{Limit Performance Result} is completed by showing that when $g < l/(1+l)$, by carefully choosing the tolerance level $K$, $GrimK$ can support cooperation shares arbitrarily close to any value between $g$ and $l/(1+l)$  in equilibrium when the survival probability is close to $1$ and the noise level is close to $0$.

\subsection*{Convergence of $GrimK$ Strategies}

Fix an arbitrary initial record distribution $\mu^{0}\in \Delta (\mathbb{N})$. When all individuals use $GrimK$ strategies, the population share with record $k$ at time $t$, $\mu_k^t$, evolves according to
\begin{equation*}
\begin{split}
\mu_{0}^{t + 1} & = 1 - \gamma + \gamma (1 - \varepsilon) \mu^{C,t} \mu_{0}^{t} , \\
\mu_{k}^{t + 1} & = \gamma (1 - (1 - \varepsilon) \mu^{C,t}) \mu_{k - 1}^{t} + \gamma (1 - \varepsilon) \mu^{C,t} \mu_{k}^{t} \text{ for } 0 < k < K , \\
\end{split}
\end{equation*}
where $\mu^{C,t}=\sum_{k=0}^{K-1}\mu_{k}^t$.

Fixing $K$, we say that distribution $\mu$ \emph{dominates} (or is \emph{more favorable than}) distribution $\tilde{\mu}$ if, for every $k < K$, $\sum_{\tilde{k}=0}^{k}\mu_{\tilde{k}}\geq \sum_{\tilde{k}=0}^{k}\tilde{\mu}_{\tilde{k}}$; that is, if for every $k < K$ the share of the population with record no worse than $k$ is greater under distribution $\mu$ than under distribution $\tilde{\mu}$. Under the $GrimK$ strategy, let $\bar{\mu}$ denote the steady state with the largest share of cooperators, and let $\underline{\mu}$ denote the steady state with the smallest share of cooperators.

Theorem \ref{convergence theorem} in the \textbf{Supplementary Information} shows that, if the initial record distribution is more favorable than $\bar{\mu}$, then the record distribution converges to $\bar{\mu}$; similarly, if the initial record distribution is less favorable that $\underline{\mu}$, then the record distribution converges to $\underline{\mu}$. Formally, if $\mu^{0}$ dominates $\bar{\mu}$, then $\lim_{t\to \infty} \mu^{t} = \bar{\mu}$; similarly, if $\mu^{0}$ is dominated by $\underline{\mu}$, then $\lim_{t\to \infty} \mu^{t} = \underline{\mu}$.

In \textbf{Fig. \ref{Convergence Figure}a} the blue trajectory corresponds to the initial distribution where all players have record $0$, the red trajectory is constant at the unique steady-state value $\mu^{C} \approx .2484$, and the yellow trajectory corresponds to the initial distribution where all players have defector records. Here  all the trajectories converge to the unique steady state. In \textbf{Fig. \ref{Convergence Figure}b}, the red trajectory is constant at the largest steady-state value $\mu^{C} \approx .9855$, the yellow trajectory is constant at the intermediate steady-state value $\mu^{C} \approx .9184$, and the purple trajectory is constant at the smallest steady-state value $\mu^{C} \approx .6471$. The blue trajectory corresponds to the initial distribution where all players have record $0$ and converges to the largest steady-state share of cooperators. The green trajectory corresponds to the initial distribution where all players have defector records and converges to the smallest steady-state share of cooperators.

\section*{Extended Data Figure}

\setcounter{figure}{0}

See \textbf{Supplementary Information} for details.

\renewcommand{\figurename}{Extended Data Figure}
\renewcommand{\thefigure}{\arabic{figure}}

\begin{figure}[H]
\centering
\begin{tikzpicture}
\node(a) at (0,4.5){
\includegraphics[width=1\textwidth]{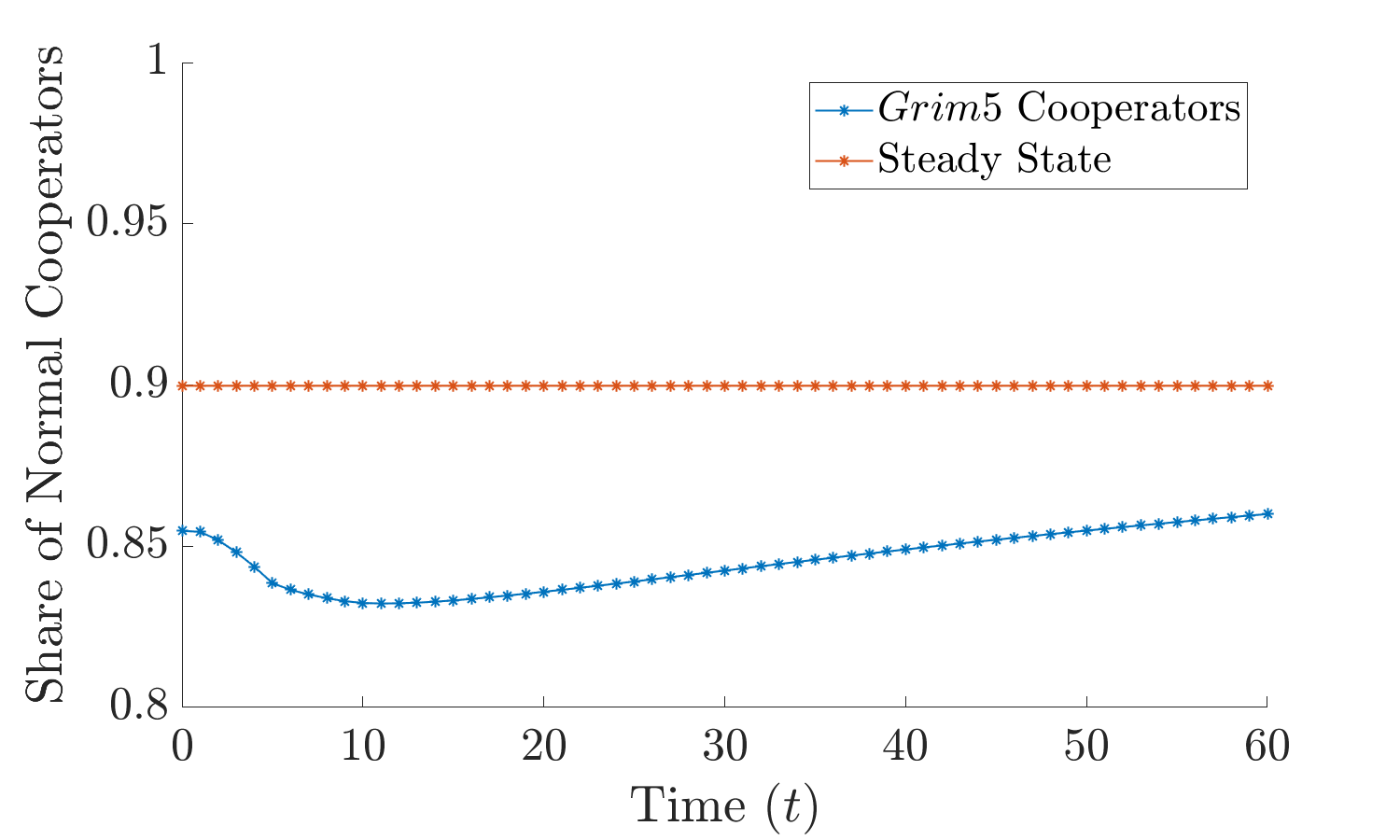}
};
\node(b) at (0,-5){
\includegraphics[width=1\textwidth]{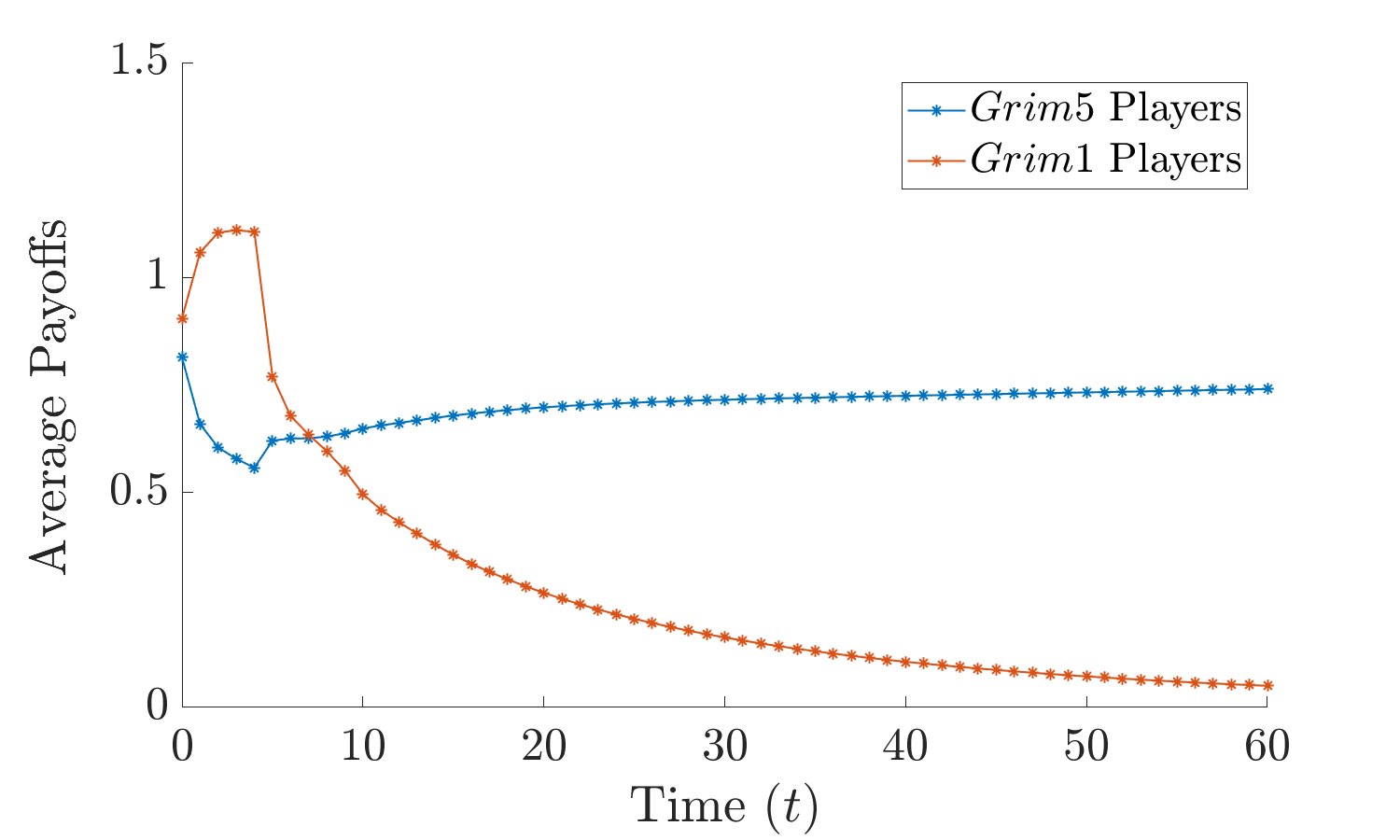}
};
\node(c) at (-8,9){\textbf{a}};
\node(d) at (-8,-.5){\textbf{b}};
\end{tikzpicture}
\caption{\textbf{Evolutionary dynamics.} \textbf{a}, The blue curve depicts the evolution of the share of players that use $Grim5$ and are cooperators (i.e. have some record $k < 5$). \textbf{b}, The average payoffs in the normal $Grim5$ population (blue curve) and in the mutant $Grim1$ population (red curve).}
\label{Dynamics Figure}
\end{figure}

\newpage

\title{Supplementary Information for ``Simple Records Support Robust Indirect Reciprocity''}

\maketitle
\tableofcontents

\setstretch{1.5}

\section{Related Work}
Work on equilibrium cooperation in repeated games began with studies of reciprocal altruism with general stage games where a fixed set of players interacts repeatedly with a commonly known start date and a common notion of calendar time [\citenum{friedmand1971supergames, fudenberg1986folk, aumann1994long}], and has been expanded to allow for various sorts of noise and imperfect observability [\citenum{Fudenberg1994TheInformation, Kandori1998PrivateCollusion, Joe2012AutomatedApproach, Compte2015, Deb2018TheMatching}]. In contrast, most evolutionary analyses of repeated games have focused on the prisoner's  dilemma [\citenum{Axelrod1981TheCooperation, boyd1987no, boyd1989evolution, farrell1989evolutionary, fudenberg1990evolution, binmore1992evolutionary, nowak1992evolutionary, nowak1992tit, nowak1993strategy, bendor1995types, axelrod2004altruism, nowak2004emergence, nowak2006five, imhof2007tit, bowles2011cooperative}], though a few evolutionary analyses have considered more complex stage games [\citenum{hauert2007via, sigmund2010social, bear2016intuition}]. 
Similarly, most laboratory and field studies of the effects of repeated interaction have also focused on the prisoner's dielmma [\citenum{Axelrod1981TheCooperation, fearon1996explaining, Fudenberg2012SlowWorld, DalBo2018OnSurvey}], though some papers consider variants with an additional third action [\citenum{dreber2008winners, rand2009positive}].

Reciprocal altruism is an important force in long-term relationships among a relatively small number players, such as business partnerships or
collusive agreements among firms, but there are many social settings where people manage to cooperate even though direct reciprocation is impossible. These interactions are better modelled as  games with repeated random matching [\citenum{Rosenthal1979}]. When the population is small compared to the discount factor, cooperation in the prisoner's dilemma can be enforced by contagion equilibria even when players have no information at all about each other's past actions [\citenum{ Kandori1992, Ellison1994, harrington1995cooperation}]. These equilibria do not exist when the population is large compared to the discount factor, so they are ruled out by our assumption of a continuum population.

Previous research on indirect reciprocity in large populations has studied the enforcement of cooperation as an equilibrium using first-order information. Takahashi [\citenum{Takahashi2010}] shows that cooperation can be supported as a strict equilibrium when the PD exhibits strategic complementarity; however, his model does not allow noise or the inflow of new players, and assumes players can use a commonly known calendar to coordinate their play. Heller and Mohlin [\citenum{Heller2018ObservationsCooperation}] show that, under strategic complementarity, the presence of a small share of players who always defect allows cooperation to be sustained as a stable (though not necessarily strict) equilibrium when players are infinitely lived and infinitely patient and are restricted to using stationary strategies. The broader importance of strategic complementarity has long been recognized in economics [\citenum{bulow1985multimarket, fudenberg1984fat}] and game theory [\citenum{vives1990nash, milgrom1990rationalizability}].

Many papers study the evolutionary selection of
cooperation using image scoring [\citenum{Sigmund2012MoralReciprocity, Nowak1998EvolutionScoring, ohtsuki2004should, santos2018social, wedekind2000cooperation, leimar2001evolution, milinski2002reputation, nowak1998dynamics, fishman2003indirect, takahashi2006importance, berger2011learning, lotem1999evolution}]. With image scoring, each player has
first-order information about their partner, but conditions their action only on
their partner's record and not on their own record. These strategies are never a
strict equilibrium, and are
typically unstable in environments with noise [\citenum{leimar2001evolution, Panchanathan2003AReciprocity}].
With more complex ``higher order" record systems such as standing, cooperation can typically be enforced in a wide range of games [\citenum{Sugden1986TheWelfare, okuno1995, Kandori1992, milinski2001cooperation, brandt2004logic, ohtsuki2004should, brandt2005indirect, pacheco2006stern, ohtsuki2007global, uchida2010competition, nakamura2011indirect}]. Most research has focused on the case where each player has only two states: for instance, Ohtsuko and Iwasa [\citenum{ohtsuki2004should,ohtsuki2006leading}] consider all possible record systems of this type, and show that only 8 of them allow an ESS with high levels of coooperation. Our first-order records can take on any  integer values, so they do not fall into this class, even though behavior is determined by a binary classification of the records. Another innovation in our model is to consider steady-state equilibria in a model with a constant inflow of new players, even without any evolutionary dynamics. This approach has previously been used to model industry dynamics in economics [\citenum{hopenhayn1992entry, jovanovic1982selection}], but is novel in the context of models of cooperation and repeated games.

The key novel aspects of our framework may thus be summarized as follows:
\begin{enumerate}
    \item Information (``records'') depends only on a player's own past actions, but players condition their behavior on their own record as well as their current partner's record.
    \item The presence of strategic complementarity implies that such two-sided conditioning can generate strict incentives for cooperation.
    \item Records are integers, and can therefore remain ``good'' even if they are repeatedly hit by noise (as is inevitable when players are long-lived).
    \item The presence of a constant inflow of new players implies that the population share with ``good'' records can remain positive even in steady state.
\end{enumerate}

\section{Model Description}

Here we formally present the model and the steady-state and equilibrium concepts.

Time is discrete and doubly infinite: $t\in \{\ldots,-2, -1,0,1,2,\ldots\}$. There is a unit mass of individuals, each with survival probability $\gamma \in (0,1)$, and an inflow of $1-\gamma$ newborns each period to keep the population size constant.

Every period, individuals randomly match in pairs to play the PD (\textbf{Fig. \ref{Prisoner's Dilemma Figure}}). Each individual carries a \emph{record} $k\in \mathbb{N} := \{0,1,2,...\}$. Newborns have record $0$. When two players meet, they observe each other's records and nothing else. A \emph{strategy} is a mapping $\mathbf{s}: \mathbb{N} \times \mathbb{N} \to \{C,D\}$. All players use the same strategy. When the players use strategy $\mathbf{s}$, the distribution over next-period records of a player with record $k$ who meets a player with record $k^{\prime}$ is given by
\begin{equation*}
\phi_{k,k^{\prime}}(\mathbf{s}) =
\begin{cases}
k \text{ w/ prob. } 1-\varepsilon, \ k+1 \text{ w/ prob. } \varepsilon & \text{if } \ \mathbf{s}(k,k')=C \\
k+1 & \text{if } \ \mathbf{s}(k,k')=D \text{ w/ prob. } 1
\end{cases}.
\end{equation*}

The \emph{state} of the system $\mu \in \Delta(\mathbb{N})$ describes the share of the population with each record, where $\mu_{k}\in [0,1]$ denotes the share with record $k$. The evolution of the state over time under strategy $\mathbf{s}$ is described by the update map $f_{\mathbf{s}}: \Delta(\mathbb{N}) \to \Delta(\mathbb{N})$, given by
\begin{equation*}
\begin{split}
f_{\mathbf{s}}(\mu )[0]&:=1-\gamma +\gamma \sum_{k'}\sum_{k''} \mu_{k'} \mu_{k''} \phi_{k',k''}(\mathbf{s})[0], \\
f_{\mathbf{s}}(\mu)[k]&:=\gamma \sum_{k'}\sum_{k''} \mu_{k'} \mu_{k''} \phi_{k',k''}(\mathbf{s})[k] \text{ for } k\neq 0 .
\end{split}
\end{equation*}
A \emph{steady state} under strategy $\mathbf{s}$ is a state $\mu$ such that $f_{\mathbf{s}}(\mu)=\mu$.

Given a strategy $\mathbf{s}$ and state $\mu$, the expected flow payoff of a player with record $k$ is $\pi_{k}(\mathbf{s},\mu)=\sum_{k'}\mu_{k'}u(\mathbf{s}(k,k'),\mathbf{s}(k',k))$, where $u$ is the (normalized) PD payoff function given by
\begin{equation*}
u(a_{1},a_{2}) =
\begin{cases}
1 & \text{if } (a_{1},a_{2}) = (C,C) \\
- l & \text{if } (a_{1},a_{2}) = (C,D) \\
1 + g & \text{if } (a_{1},a_{2}) = (D,C) \\
0 & \text{if } (a_{1},a_{2}) = (D,D)
\end{cases} .
\end{equation*}
Denote the probability that a player with current record $k$ has record $k'$ $t$ periods in the future by $\phi_{k}(\mathbf{s},\mu)^{t}(k')$. The continuation payoff of a player with record $k$ is then $V_{k}(\mathbf{s},\mu) = (1-\gamma)\sum_{t=0}^{\infty}\gamma^{t}\sum_{k'}\phi_{k}(\mathbf{s},\mu)^{t}(k')\pi_{k'}(\mathbf{s},\mu).$ A player's objective is to maximize their expected lifetime payoff.

A pair $(\mathbf{s},\mu)$ is an \textit{equilibrium} if $\mu$ is a steady-state under $\mathbf{s}$ and, for each own record $k$ and opponent's record $k'$, $\mathbf{s}(k,k') \in \{C,D\}$ maximizes $(1-\gamma)u(a,\mathbf{s}(k',k)) + \gamma \sum_{k''}\left( \rho(k,a)[k'']\right) V_{k''}(\mathbf{s},\mu)$ over $a\in\{C,D\}$, where $\rho(k,a)[k'']$ denotes the probability that a player with record $k$ who takes action $a$ acquires next-period record $k''$. An equilibrium is \emph{strict} if the maximizer is unique for all pairs $(k,k')$. 

This equilibrium definition encompasses two forms of strategic robustness. First, we allow agents to maximize over all possible strategies, as opposed to only strategies from some pre-selected set. Second, we focus on strict equilibria, which remain equilibria under ``small" perturbations of the model.

\section{Limit Cooperation under $GrimK$ Strategies}

Under $GrimK$ strategies, a matched pair of players cooperate if and only if both records are below a pre-specified cutoff $K$: that is, $s(k,k')=C$ if $\max\{k,k'\}<K$ and $s(k,k')=D$ if $\max\{k,k'\}\geq K$.  

We call an individual a \emph{cooperator} if their record is below $K$ and a \emph{defector} otherwise. Note that each individual may be a cooperator for some periods of their life and a defector for other periods.

Given an equilibrium strategy $GrimK$, let $\mu^{C}=\sum_{k=0}^{K-1}\mu_{k}$ denote the corresponding steady-state share of cooperators. Note that, in a steady state with cooperator share $\mu^{C}$, mutual cooperation is played in share $(\mu^{C})^2$ of all matches. Let $\overline{\mu}^{C}(\gamma, \varepsilon)$ be the maximal share of cooperators in any tolerant trigger equilibrium (allowing for every possible $K$) when the survival probability is $\gamma$ and the noise level is $\varepsilon$.

The following theorem characterizes the performance of equilibria in $GrimK$ strategies in the double limit of interest [\citenum{ohtsuki2004should, ohtsuki2006leading, Ellison1994, horner2006folk, Takahashi2010}] where the survival probability approaches $1$---so that players expect to live a long time and the ``shadow of the future" looms large---and the noise level approaches $0$---so that players who play $C$ are unlikely to be recorded as playing $D$.

\begin{theorem} 
\begin{equation*}
\lim_{(\gamma, \varepsilon) \rightarrow (1,0)} \overline{\mu}^{C}(\gamma, \varepsilon) =
\begin{cases}
\frac{l}{1 + l} & \text{if } g < \frac{l}{1 + l} \\
0 & \text{if } g > \frac{l}{1 + l}
\end{cases} .
\end{equation*}
\label{Limit Performance Result}
\end{theorem}

To prove the theorem, let $\beta: (0,1) \times (0,1) \times (0,1) \rightarrow (0,1)$ be the function given by
\begin{equation}
\beta(\gamma, \varepsilon, \mu^{C}) = \frac{\gamma (1 - (1 - \varepsilon) \mu^{C})}{1 - \gamma (1 - \varepsilon) \mu^{C}} .
\label{beta Expression}
\end{equation}
When players use $GrimK$ strategies and the share of cooperators is $\mu^C$, $\beta(\gamma,\varepsilon,\mu^C)$ is the probability that a player with cooperator record $k$ survives to reach record $k+1$. (This probability is the same for all $k<K$.)

\begin{lemma}
There is a $GrimK$ equilibrium with cooperator share $\mu^{C}$ if and only if the following conditions hold:
\begin{enumerate}
    \item Feasibility:
    \begin{equation}
        \mu^{C} = 1 - \beta(\gamma, \varepsilon, \mu^{C})^{K} . \label{feasible}
    \end{equation}
    
    \item Incentives:
    \begin{align}
        \frac{(1 - \varepsilon) (1 - \mu^{C})}{1 - (1 - \varepsilon) \mu^{C}} \mu^{C} > g , \label{ICC} \\
        \mu^{C} < \frac{1}{\gamma (1 - \varepsilon)} \frac{l}{1 + l} . \label{ICD}
\end{align}
\end{enumerate}
\label{Equilibrium Characterization}
\end{lemma}

Note that $\mu^{C} = 0$ solves (\ref{feasible}) when $K = 0$. For any $K > 0$, $0 < 1 - \beta(\gamma,\varepsilon,\mu^{C})^{K}$ and $1 > 1 - \beta(\gamma,\varepsilon,1)^{K}$, so by the intermediate value theorem, (\ref{feasible}) has some solution $\mu \in (0,1)$. Thus, there is at least one steady state for every $GrimK$ strategy. For some strategies, there are multiple steady states, but never more than $K+1$, because  (\ref{feasible}) can be rewritten as a polynomial equation in $\mu^{C}$ with degree $K + 1$.

The upper bounds on the equilibrium share of cooperators in Table \ref{Upper Bounds on Cooperation Table} are the suprema of the $\mu^{C} \in (0,1)$ that satisfy (\ref{ICC}) and (\ref{ICD}) for the corresponding $(\gamma, \varepsilon)$ parameters. When no $\mu^{C} \in (0,1)$ satisfy (\ref{ICC}) and (\ref{ICD}), the upper bound is $0$, since $Grim0$ (where everyone plays $D$) is always a strict equilibrium.

To see how Part 2 of Theorem \ref{Limit Performance Result} comes from Lemma \ref{Equilibrium Characterization}, note that
\begin{equation*}
\frac{(1 - \varepsilon) (1 - \mu^{C})}{1 - (1 - \varepsilon) \mu^{C}} \leq 1 .
\end{equation*}
Thus, (\ref{ICC}) requires $\mu^{C} > g$. Moreover, combining $\mu^{C} > g$ with (\ref{ICD}) gives $\gamma (1 - \varepsilon) g < l / (1 + l)$. Taking the $(\gamma, \varepsilon) \rightarrow (1,0)$ limit of this inequality gives $g \leq l / (1 + l)$. Thus, when $g > l / (1 + l)$, it follows that $\lim_{(\gamma, \varepsilon) \rightarrow (1,0)} \overline{\mu}^{C}(\gamma, \varepsilon) = 0$.

All that remains is to show that $\lim_{(\gamma, \varepsilon) \rightarrow (1,0)} \overline{\mu}^{C}(\gamma, \varepsilon) = l / (1 + l)$ when $g > l/(1 + l)$. Since $\lim_{\varepsilon \rightarrow 0} (1 - \varepsilon) (1 - \mu^{C}) / (1 - (1 - \varepsilon) \mu^{C}) = 1$ for any fixed $\mu^{C}$ and $\lim_{(\gamma, \varepsilon) \rightarrow (1,0)} 1 / (\gamma (1 - \varepsilon)) = 1$, it follows that values of $\mu^{C}$ smaller than, but arbitrarily close to, $l / (1 + l)$ satisfy (\ref{ICC}) and (\ref{ICD}) in the double limit. Thus, the only difficulty is showing the feasibility of $\mu^{C}$ as a steady-state level of cooperation: because $K$ must be an integer, some values of $\mu^{C}$ cannot be generated by any $K$, for given values of $\gamma$ and $\varepsilon$. The following result shows that this ``integer problem'' becomes irrelevant in the limit. That is, any value of $\mu^{C} \in (0,1)$ can be approximated arbitrarily closely by a feasible steady-state share of cooperators for some $GrimK$ strategy as $(\gamma,\varepsilon) \rightarrow (1,0)$.

\begin{lemma}
Fix any $\mu^{C} \in (0,1)$. For all $\Delta > 0$, there exist $\overline{\gamma} < 1$ and $\overline{\varepsilon} > 0$ such that, for all $\gamma > \overline{\gamma}$ and $\varepsilon < \overline{\varepsilon}$, there exists $\hat{\mu}^{C}$ that satisfies (\ref{feasible}) for some $K$ such that $|\hat{\mu}^{C} - \mu^{C}| < \Delta$.
\label{Denseness Result}
\end{lemma}

\subsubsection*{Proof of Lemma \ref{Equilibrium Characterization}}

The feasibility condition of Lemma \ref{Equilibrium Characterization} comes from the following lemma.
\begin{lemma}
In a $GrimK$ equilibrium with cooperator share $\mu^{C}$, $\mu_{k} = \beta(\gamma, \varepsilon, \mu^{C})^{k} (1 - \beta(\gamma, \varepsilon, \mu^{C}))$ for all $k < K$.
\label{Steady State Shares Lemma}
\end{lemma}

To see why the feasibility condition of Lemma \ref{Equilibrium Characterization} comes from Lemma \ref{Steady State Shares Lemma}, note that
\begin{equation*}
\mu^{C} = \sum_{k = 0}^{K - 1} \beta(\gamma, \varepsilon, \mu^{C})^{k} (1 - \beta(\gamma, \varepsilon, \mu^{C})) = 1 - \beta(\gamma, \varepsilon, \mu^{C})^{K} .
\end{equation*}

\begin{proof}[Proof of Lemma \ref{Steady State Shares Lemma}]

The inflow into record $0$ is $1 - \gamma$, while the outflow from record $0$ is $(1 - \gamma (1 - \varepsilon) \mu^{C}) \mu_{0}$. Setting these equal gives
\begin{equation*}
\mu_{0} = \frac{1 - \gamma}{1 - \gamma (1 - \varepsilon) \mu^{C}} = 1 - \beta(\gamma, \varepsilon, \mu^{C}) .
\end{equation*}
Additionally, for every $0 < k < K$, the inflow into record $k$ is $\gamma (1 - (1 - \varepsilon) \mu^{C}) \mu_{k - 1}$, while the outflow from record $k$ is $(1 - \gamma (1 - \varepsilon) \mu^{C}) \mu_{k}$. Setting these equal gives
\begin{equation*}
\mu_{k} = \frac{\gamma (1 - (1 - \varepsilon) \mu^{C})}{1 - \gamma (1 - \varepsilon) \mu^{C}} \mu_{k - 1} = \beta(\gamma, \varepsilon, \mu^{C}) \mu_{k - 1} .
\end{equation*}
Combining this with $\mu_{0} = 1 - \beta(\gamma, \varepsilon, \mu^{C})$ gives $\mu_{k} = \beta(\gamma, \varepsilon, \mu^{C})^{k} (1 - \beta(\gamma, \varepsilon, \mu^{C}))$ for $0 \leq k \leq K - 1$. \end{proof}

The incentive constraint (\ref{ICC}) of Lemma \ref{ICD} guarantees that a record-$0$ cooperator plays $C$ against an opponent playing $C$, and the incentive constraint (\ref{Equilibrium Characterization}) guarantees that a record-$(K - 1)$ cooperator plays $D$ against an opponent playing $D$. Record-$0$ cooperators are the cooperators most tempted to defect against a cooperative opponent and record-$(K - 1)$ cooperators are the cooperators most tempted to cooperate against a defecting opponent, so these constraints guarantee the incentives of all cooperators are satisfied. 
\begin{lemma}
In a $GrimK$ equilibrium with cooperator share $\mu^{C}$,
\begin{equation*}
V_{k} = 
\begin{cases}
(1 - \beta(\gamma, \varepsilon, \mu^{C})^{K - k}) \mu^{C} & \text{if } k < K \\
0 & \text{if } k \geq K .
\end{cases}
\end{equation*}
\label{Value Functions Lemma}
\end{lemma}

To see why the incentive constraints of Lemma \ref{Equilibrium Characterization} come from Lemma \ref{Value Functions Lemma}, note that the expected continuation payoff of a record-$0$ player from playing $C$ is $(1 - \varepsilon) V_{0} + \varepsilon V_{1}$, while the expected continuation payoff from playing $D$ is $V_{1}$. Thus, a record $0$ player strictly prefers to play $C$ against an opponent playing $C$ iff $(1 - \varepsilon) \gamma (V_{0} - V_{1}) / (1 - \gamma) > g$. Combining Lemmas \ref{Steady State Shares Lemma} and \ref{Value Functions Lemma} gives
\begin{equation*}
(1 - \varepsilon) \frac{\gamma}{1 - \gamma} (V_{0} - V_{1}) = \frac{1 - \varepsilon}{1 - (1 - \varepsilon) \mu^{C}} \beta(\gamma, \varepsilon, \mu^{C})^{K} \mu^{C} = \frac{(1 - \varepsilon) (1 - \mu^{C})}{1 - (1 - \varepsilon) \mu^{C}} \mu^{C} ,
\end{equation*}
so (\ref{ICC}) follows. Moreover, the expected continuation payoff of a record $K - 1$ player from playing $C$ is $(1 - \varepsilon) V_{K - 1} + \varepsilon V_{K}$, while the expected continuation payoff from playing $D$ is $V_{K}$. Thus, a record $K - 1$ player strictly prefers to play $D$ against an opponent playing $D$ iff $(1 - \varepsilon) \gamma (V_{K - 1} - V_{K}) / (1 - \gamma) < l$. Lemma \ref{Value Functions Lemma} gives
\begin{equation*}
(1 - \varepsilon) \frac{\gamma}{1 - \gamma} (V_{K - 1} - V_{K}) = \frac{\gamma (1 - \varepsilon) \mu^{C}}{1 - \gamma (1 - \varepsilon) \mu^{C}} ,
\end{equation*}
and setting this to be less than $l$ gives (\ref{ICD}).

\begin{proof}[Proof of Lemma \ref{Value Functions Lemma}]

The flow payoff for any record $k \geq K$ is $0$, so $V_{k} = 0$ for $k \geq K$. For $k < K$, $V_{k} = (1 - \gamma) \mu^{C} + \gamma (1 - \varepsilon) \mu^{C} V_{k} + \gamma (1 - (1 - \varepsilon) \mu^{C}) V_{k + 1}$, which gives $V_{k} = (1 - \beta(\gamma, \varepsilon, \mu^{C})) \mu^{C} + \beta(\gamma, \varepsilon, \mu^{C}) V_{k + 1}$. Combining this with $V_{K} = 0$ gives $V_{k} = (1 - \beta(\gamma, \varepsilon, \mu^{C})^{K - k}) \mu^{C}$ for $k < K$. \end{proof}

\subsubsection*{Proof of Lemma \ref{Denseness Result}}

Let $\tilde{K}: (0, 1) \times (0,1) \times (0, 1) \rightarrow \mathbb{R}_{+}$ be the function given by
\begin{equation}
\tilde{K}(\gamma, \varepsilon, \mu^{C}) = \frac{\ln(1 - \mu^{C})}{\ln(\beta(\gamma, \varepsilon, \mu^{C}))} .
\label{Defining K Equation}
\end{equation}
By construction, $\tilde{K}(\gamma, \varepsilon, \mu^{C})$ is the unique $K \in \mathbb{R}_{+}$ such that $\mu^{C} = 1 - \beta(\gamma, \varepsilon, \mu^{C})^{K}.$ Let $d: (0,1] \times [0,1) \times (0,1) \rightarrow \mathbb{R}$ be the function given by
\begin{equation*}
d(\gamma, \varepsilon, \mu^{C}) =
\begin{cases}
1 + \ln(1 - \mu^{C}) (1 - \mu^{C}) \frac{\frac{\partial \beta}{\partial \mu^{C}}(\gamma, \varepsilon, \mu^{C})}{\beta(\gamma, \varepsilon, \mu^{C}) \ln(\beta(\gamma, \varepsilon, \mu^{C}))} & \text{if } \gamma < 1 \\
1 + \frac{(1 - \varepsilon) \ln(1 - \mu^{C}) (1 - \mu^{C})}{1 - (1 - \varepsilon) \mu^{C}} & \text{if } \gamma = 1
\end{cases} .
\end{equation*}

The $\mu^{C}$ derivative of $\tilde{K}(\gamma, \varepsilon, \mu^{C})$ is related to $d(\gamma, \varepsilon, \mu^{C})$ by the following lemma.
\begin{lemma}
$\tilde{K}: (0,1) \times (0,1) \times (0,1) \rightarrow \mathbb{R}_{+}$ is differentiable in $\mu^{C}$ with derivative given by
\begin{equation*}
\frac{\partial \tilde{K}}{\partial \mu^{C}}(\gamma, \varepsilon, \mu^{C}) = - \frac{d(\gamma, \varepsilon, \mu^{C})}{(1 - \mu^{C}) \ln(\beta(\gamma, \varepsilon, \mu^{C}))} .
\end{equation*}
\label{Derivative Result}
\end{lemma}

\begin{proof}[Proof of Lemma \ref{Derivative Result}]
From (\ref{Defining K Equation}), it follows that $\tilde{K}(\gamma, \varepsilon, \mu^{C})$ is differentiable in $\mu^{C}$ with derivative given by

\begin{align*}
\frac{\partial \tilde{K}}{\partial \mu^{C}}(\gamma, \varepsilon, \mu^{C}) & = - \frac{\frac{\ln(\beta(\gamma, \varepsilon, \mu^{C}))}{1 - \mu^{C}} + \frac{\ln(1 - \mu^{C}) \frac{\partial \beta}{\partial \mu^{C}}(\gamma, \varepsilon, \mu^{C})}{\beta(\gamma, \varepsilon, \mu^{C})}}{\ln(\beta(\gamma, \varepsilon, \mu^{C}))^{2}} \\
& = - \frac{1 + \ln(1 - \mu^{C}) (1 - \mu^{C}) \frac{\frac{\partial \beta}{\partial \mu^{C}}(\gamma, \varepsilon, \mu^{C})}{\beta(\gamma, \varepsilon, \mu^{C}) \ln(\beta(\gamma, \varepsilon, \mu^{C}))}}{(1 - \mu^{C}) \ln(\beta(\gamma, \varepsilon, \mu^{C}))} \\
& = - \frac{d(\gamma, \varepsilon, \mu^{C})}{(1 - \mu^{C}) \ln(\beta(\gamma, \varepsilon, \mu^{C}))} . \qedhere
\end{align*}
\end{proof}

The following two lemmas concern properties of $d(\gamma, \varepsilon, \mu^{C})$ that will be useful for the proof of Lemma \ref{Denseness Result}.

\begin{lemma}
$d: (0,1] \times [0,1) \times (0,1) \rightarrow \mathbb{R}$ is well-defined and continuous.
\label{Continuity Result}
\end{lemma}

\begin{proof}[Proof of Lemma \ref{Continuity Result}]

Since $\beta(\gamma, \varepsilon, \mu^{C})$ is differentiable and only takes values in $(0,1)$, it follows that $d(\gamma, \varepsilon, \mu^{C})$ is well-defined. Moreover, since $\beta(\gamma, \varepsilon, \mu^{C})$ is continuously differentiable for all $\mu^{C} \in (0,1)$, $d(\gamma, \varepsilon, \mu^{C})$ is continuous for $\gamma < 1$. All that remains is to check that $d(\gamma, \varepsilon, \mu^{C})$ is continuous for $\gamma = 1$. 

First, note that $d(1, \varepsilon, \mu^{C})$ is continuous in $(\varepsilon, \mu^{C})$. Thus, we need only check the limit in which $\gamma$ approaches 1, but never equals $1$. Note that
\begin{equation}
\begin{split}
\frac{\frac{\partial \beta}{\partial \mu^{C}}(\gamma, \varepsilon, \mu^{C})}{\beta(\gamma, \varepsilon, \mu^{C}) \ln(\beta(\gamma, \varepsilon, \mu^{C}))} & = - \frac{\frac{\gamma (1 - \varepsilon) (1 - \gamma)}{(1 - \gamma (1 - \varepsilon) \mu^{C})^{2}}}{\beta(\gamma, \varepsilon, \mu^{C}) \ln(\beta(\gamma, \varepsilon, \mu^{C}))} \\
& = - \Big( \frac{\gamma (1 - \varepsilon)}{\beta(\gamma, \varepsilon, \mu^{C}) (1 - \gamma(1 - \varepsilon) \mu^{C})} \Big) \Big( \frac{1 - \beta(\gamma, \varepsilon, \mu^{C})}{\ln(\beta(\gamma, \varepsilon, \mu^{C}))} \Big) .
\end{split}
\label{Continuity Equation}
\end{equation}
It is clear that
\begin{equation}
\lim_{\substack{(\tilde{\gamma}, \tilde{\varepsilon}, \tilde{\mu}^{C}) \rightarrow (1, \varepsilon, \mu^{C}) \\ \tilde{\gamma} \neq 1}} \frac{\tilde{\gamma} (1 - \tilde{\varepsilon})}{\beta(\tilde{\gamma}, \tilde{\varepsilon}, \tilde{\mu^{C}}) (1 - \tilde{\gamma}(1 - \tilde{\varepsilon}) \tilde{\mu}^{C})} = \frac{1 - \varepsilon}{(1 - (1 - \varepsilon) \mu^{C})}
\label{Limit Equation 1}
\end{equation}
for all $(\varepsilon, \mu^{C}) \in [0,1) \times (0,1)$. For $\gamma$ close to $1$,
\begin{equation*}
\ln(\beta(\gamma, \varepsilon, \mu^{C})) = \beta(\gamma, \varepsilon, \mu^{C}) - 1 + O((\beta(\gamma, \varepsilon, \mu^{C}) - 1)^{2}) .
\end{equation*}
Thus,
\begin{equation}
\lim_{\substack{(\tilde{\gamma}, \tilde{\varepsilon}, \tilde{\mu}) \rightarrow (1, \varepsilon, \mu^{C}) \\ \tilde{\gamma} \neq 1}} \frac{1 - \beta(\tilde{\gamma}, \tilde{\varepsilon}, \tilde{\mu}^{C})}{\ln(\beta(\tilde{\gamma}, \tilde{\varepsilon}, \tilde{\mu}^{C}))} = - 1
\label{Limit Equation 2}
\end{equation}
for all $(\varepsilon, \mu^{C}) \in [0,1) \times (0,1)$. Equations \ref{Continuity Equation}, \ref{Limit Equation 1}, and \ref{Limit Equation 2} together imply that $d(\gamma, \varepsilon, \mu^{C})$ is continuous for $\gamma = 1$.
\end{proof}

\begin{lemma}
$d(1, 0, \mu^{C})$ has precisely one zero in $\mu^{C} \in (0,1)$, and the zero is located at $\mu^{C} = 1 - 1/e$.
\label{Zeros Result}
\end{lemma}

\begin{proof}[Proof of Lemma \ref{Zeros Result}]

This follows from the fact that $d(1,0,\mu^{C}) = 1 + \ln(1 - \mu^{C})$.
\end{proof}

With these preliminaries established, we now present the proof of Lemma \ref{Denseness Result}.

\begin{proof}[Proof of Lemma \ref{Denseness Result}]
Fix some $\tilde{\mu}^{C} \in (0,1)$ such that $\tilde{\mu}^{C} \neq 1 - 1 / e$. Lemma \ref{Zeros Result} says $d(1, 0, \tilde{\mu}^{C}) \neq 0$. Because of this and the continuity of $d$, there exist some $\lambda > 0$ and some $\delta > 0$, $\overline{\gamma}' < 1$, and $\overline{\varepsilon} > 0$ such that $|d(\gamma, \varepsilon, \mu^{C})| > \lambda$ for all $\gamma > \overline{\gamma}'$, $\varepsilon < \overline{\varepsilon}$, and $|\mu^{C} - \tilde{\mu}^{C}| < \delta$.

Additionally, note that $\lim_{\gamma \rightarrow 1} \inf_{(\varepsilon, \mu^{C}) \in (0,\overline{\varepsilon}) \times (\mu^{C} - \delta,\mu^{C} + \delta)} \beta(\gamma, \varepsilon, \mu^{C}) = 1$. Together these facts imply that there exists some $\overline{\gamma} < 1$ such that
\begin{equation*}
\left|\frac{\partial \tilde{K}}{\partial \mu^{C}}(\gamma, \varepsilon, \mu^{C})\right| = \left|\frac{d(\gamma, \varepsilon, \mu^{C})}{(1 - \mu^{C}) \ln(\beta(\gamma, \varepsilon, \mu^{C}))}\right| > \frac{2}{\min\{\delta, \Delta\}}
\end{equation*}
and $\tilde{K}(\gamma, \varepsilon, \mu^{C}) \geq 1$ for all $\gamma > \overline{\gamma}$, $\varepsilon < \overline{\varepsilon}$, and $|\mu^{C} - \tilde{\mu}^{C}| < \delta$. It thus follows that 
\begin{equation*}
\sup_{|\mu^{C} - \tilde{\mu}^{C}| \leq \min\{\delta, \Delta\}} |\tilde{K}(\gamma, \varepsilon, \mu^{C}) - \tilde{K}(\gamma, \varepsilon, \tilde{\mu}^{C})| > 1
\end{equation*}
for all $\gamma > \overline{\gamma}$, $\varepsilon < \overline{\varepsilon}$. Hence, there exists some $\hat{\mu}^{C}$ within $\Delta$ of $\tilde{\mu}^{C}$ and some non-negative integer $\hat{K}$ such that $\tilde{K}(\gamma, \varepsilon, \hat{\mu}^{C}) = \hat{K}$, which implies that $\hat{\mu}^{C}$ is feasible since $\hat{\mu}^{C} = 1 - \beta(\gamma, \varepsilon, \hat{\mu}^{C})^{\hat{K}}$.
\end{proof}

\section{Convergence of $GrimK$ Strategies}

We now derive a key stability property of $GrimK$ strategies. Fix an arbitrary initial record distribution $\mu^{0}\in \Delta (\mathbb{N})$. When all individuals use $GrimK$ strategies, the population share with record $k$ at time $t$, $\mu_k^t$, evolves according to
\begin{equation}
\begin{split}
\mu_{0}^{t + 1} & = 1 - \gamma + \gamma (1 - \varepsilon) \mu^{C,t} \mu_{0}^{t} , \\
\mu_{k}^{t + 1} & = \gamma (1 - (1 - \varepsilon) \mu^{C,t}) \mu_{k - 1}^{t} + \gamma (1 - \varepsilon) \mu^{C,t} \mu_{k}^{t} \text{ for } 0 < k < K , \\
\end{split}
\label{Governing Dynamics}
\end{equation}
where $\mu^{C,t}=\sum_{k=0}^{K-1}\mu_{k}^t$.

Fixing $K$, we say that distribution $\mu$ \emph{dominates} (or is \emph{more favorable than}) distribution $\tilde{\mu}$ if, for every $k < K$, $\sum_{\tilde{k}=0}^{k}\mu_{\tilde{k}}\geq \sum_{\tilde{k}=0}^{k}\tilde{\mu}_{\tilde{k}}$; that is, if for every $k < K$ the share of the population with record no worse than $k$ is greater under distribution $\mu$ than under distribution $\tilde{\mu}$. Under the $GrimK$ strategy, let $\bar{\mu}$ denote the steady state with the largest share of cooperators, and let $\underline{\mu}$ denote the steady state with the smallest share of cooperators.

\begin{theorem} ~
\begin{enumerate}
    
    \item If $\mu^{0}$ dominates $\bar{\mu}$, then $\lim_{t\to \infty} \mu^{t} = \bar{\mu}.$
    
    \item If $\mu^{0}$ is dominated by $\underline{\mu}$, then $\lim_{t\to \infty} \mu^{t} = \underline{\mu}.$
    
\end{enumerate} \label{convergence theorem}
\end{theorem}

Let $x_{k}=\sum_{\tilde{k}=0}^{k}\mu_{\tilde{k}}$ denote the share of the population with record no worse than $k$. From Equation \ref{Governing Dynamics}, it follows that
\begin{equation}
\begin{split}
x_{0}^{t + 1} & = 1 - \gamma + \gamma (1 - \varepsilon) x_{K - 1}^{t} x_{0}^{t} , \\
x_{k}^{t + 1} & = 1 - \gamma + \gamma x_{k - 1}^{t} + \gamma (1 - \varepsilon) x_{K - 1}^{t} (x_{k}^{t} - x_{k - 1}^{t}) \text{ for } 0 < k < K . \\
\end{split}
\label{Updated Governing Dynamics}
\end{equation}
To see this, note that $x_{0} = \mu_{0}$ and $x_{K - 1} = \mu^{C}$, so rewriting the first line in Equation \ref{Governing Dynamics} gives the first line in Equation \ref{Updated Governing Dynamics}. Additionally, for $0 < k < K$, Equation \ref{Governing Dynamics} gives
\begin{equation*}
\begin{split}
x_{k}^{t + 1} = \sum_{\tilde{k} \leq k} \mu_{\tilde{k}}^{t + 1} & = 1 - \gamma + \gamma \sum_{\tilde{k} \leq k - 1} \mu_{\tilde{k} - 1}^{t} + \gamma (1 - \varepsilon) \mu^{C,t} \mu_{k}^{t} , \\
& = 1 - \gamma + \gamma x_{k - 1}^{t} + \gamma (1 - \varepsilon) x_{K - 1}^{t} (x_{k}^{t} - x_{k - 1}^{t}) .
\end{split}
\end{equation*}

\begin{lemma}
The update map in Equation \ref{Updated Governing Dynamics} is weakly increasing: If $(x_{0}^{t}, ..., x_{K - 1}^{t}) \geq (\tilde{x}_{0}^{t}, ..., \tilde{x}_{K - 1}^{t})$, then $(x_{0}^{t + 1}, ..., x_{K - 1}^{t + 1}) \geq (\tilde{x}_{0}^{t + 1}, ..., \tilde{x}_{K - 1}^{t + 1})$.
\label{Update Function Monotonicity}
\end{lemma}

\begin{proof}[Proof of Lemma \ref{Update Function Monotonicity}]

The right-hand side of the first line in Equation \ref{Updated Governing Dynamics} depends only on the product of $x_{0}^{t}$ and $x_{K - 1}^{t}$, and it is strictly increasing in this product. The right-hand side of the second line in Equation \ref{Updated Governing Dynamics} depends only on $x_{k - 1}^{t}$, $x_{k}^{t}$, and $x_{K - 1}^{t}$, and, holding fixed any two of these variables, it is weakly increasing in the third variable.
\end{proof}

\begin{proof}[Proof of Theorem \ref{convergence theorem}]

We prove the first statement of Theorem \ref{convergence theorem}. A similar argument handles the second statement. Let $(\tilde{x}_{0}^{t},...,\tilde{x}_{K - 1}^{t})$ denote the time-path corresponding to the highest possible initial conditions, i.e. $(\tilde{x}_{0}^{0},...,\tilde{x}_{K - 1}^{0}) = (1,...,1)$. By Lemma \ref{Update Function Monotonicity}, $(\tilde{x}_{0}^{t + 1},...,\tilde{x}_{K - 1}^{t + 1}) \leq (\tilde{x}_{0}^{t},...,\tilde{x}_{K - 1}^{t})$ for all $t$. Thus, it follows that $\lim_{t \rightarrow \infty} (\tilde{x}_{0}^{t},...,\tilde{x}_{K - 1}^{t}) = \inf_{t} (\tilde{x}_{0}^{t},...,\tilde{x}_{K - 1}^{t})$, so in particular $\lim_{t \rightarrow \infty} (\tilde{x}_{0}^{t},...,\tilde{x}_{K - 1}^{t})$ exists. Since the update rules in Equation \ref{Updated Governing Dynamics} are continuous, it follows that $\lim_{t \rightarrow \infty} (\tilde{x}_{0}^{t},...,\tilde{x}_{K - 1}^{t})$ must be a steady state of the system. By Lemma \ref{Update Function Monotonicity} and the fact that $(\overline{x}_{0},...,\overline{x}_{K - 1})$ is the steady state with the highest share of cooperators, it follows that $\lim_{t \rightarrow \infty} (\tilde{x}_{0}^{t},...,\tilde{x}_{K - 1}^{t}) = (\overline{x}_{0},...,\overline{x}_{K - 1})$.

Now, fix some $(x_{0}^{0},...,x_{K - 1}^{0}) \geq (\overline{x}_{0},...,\overline{x}_{K - 1})$. By Lemma \ref{Update Function Monotonicity},
\begin{equation*}
(\overline{x}_{0},...,\overline{x}_{K - 1}) \leq (x_{0}^{t},...,x_{K - 1}^{t}) \leq (\tilde{x}_{0}^{t},...,\tilde{x}_{K - 1}^{t})
\end{equation*}
for all $t$, so it follows that $\lim_{t \rightarrow \infty} (x_{0}^{t},...,x_{K - 1}^{t}) = (\overline{x}_{0},...,\overline{x}_{K - 1})$.
\end{proof}

\section{Evolutionary Analysis}

We have so far analyzed the efficiency of $GrimK$ equilibrium steady states (Theorem \ref{Limit Performance Result}) and convergence to such steady states when all players use the $GrimK$ strategy (Theorem \ref{convergence theorem}). To further examine the robustness of $GrimK$ strategies, we now perform two types of evolutionary analysis. In Section \ref{section SS robustness}, we show that, when $g < l/(1 + l)$, there are sequences of $GrimK$ equilibria that obtain the maximum cooperator share of $l / (1 + l)$ as  $(\gamma,\varepsilon) \rightarrow (1,0)$ that are robust to  invasion by a small mass of mutants who follow any other $GrimK'$ strategy, such as \emph{Always Defect} (i.e., $Grim0$). In Section \ref{section evolutionary dynamics}, we report simulations of the evolutionary dynamic when a $GrimK$ steady state is invaded by mutants playing another $GrimK'$ strategy.

\subsection{Steady-State Robustness} \label{section SS robustness}

We consider the following notion of steady-state robustness.

\begin{definition}
A $GrimK$ equilibrium with share of cooperators $\mu^{C}$ is \textbf{steady-state robust to mutants} if, for every $K' \neq K$ and $\alpha > 0$, there exists some $\overline{\delta} > 0$ such that when the share of players playing $GrimK$ is $1 - \delta$ and the share of players playing $GrimK'$ is $\delta$ with $\delta < \overline{\delta}$, then
\begin{itemize}

\item There is a steady state where the fraction of players playing $GrimK$ that are cooperators, $\tilde{\mu}^{C}$, satisfies $|\tilde{\mu}^{C} - \mu^{C}| < \alpha$, and

\item It is strictly optimal to play $GrimK$.

\end{itemize}
\end{definition}

We show that, whenever strategic complementarities are strong enough to support a cooperative $GrimK$ equilibrium, there is a sequence of $GrimK$ equilibria that are robust to mutants and attains  the maximum cooperation level of $l / (1 + l)$ when expected lifespans are long and noise is small.

\begin{theorem}
Suppose that $g < l / (1 + l)$. There is a family of $GrimK$ equilibria giving a share of cooperators $\mu^{C}(\gamma,\varepsilon)$ for parameters $\gamma,\varepsilon$ such that:
\begin{enumerate}

\item $\lim_{(\gamma,\varepsilon) \rightarrow (1,0)} \mu^{C}(\gamma,\varepsilon) = l / (1 + l)$, and

\item There is some $\overline{\gamma} < 1$ and $\overline{\varepsilon} > 0$ such that, when $\gamma > \overline{\gamma}$ and $\varepsilon < \overline{\varepsilon}$, the $GrimK$ equilibrium with share of cooperators $\overline{\mu}^{C}(\gamma,\varepsilon)$ is steady-state robust to mutants.

\end{enumerate}
\label{Evolutionary Stability Result}
\end{theorem}

\begin{proof}
We assume that $K' < K$; the proof for $K' > K$ is analogous. Fix some $g < \tilde{\mu}^{C} < l / (1 + l)$ satisfying $\tilde{\mu}^{C} \neq 1 - 1/e$. By Lemmas \ref{Equilibrium Characterization} and \ref{Denseness Result}, we know that there exists some family of tolerant trigger strategy equilibria with share of cooperators $\tilde{\mu}^{C}(\gamma,\varepsilon)$ such that $\lim_{(\gamma,\varepsilon) \rightarrow (1,0)} \tilde{\mu}^{C}(\gamma,\varepsilon) = \tilde{\mu}^{C}$. Fix some $\gamma,\varepsilon$, and consider the modified environment where share $1 - \delta$ of the players use the $GrimK$ strategy corresponding to $\tilde{\mu}^{C}(\gamma,\varepsilon)$ and share $\delta$ of the players use some other $GrimK'$.

 Let  $\mu^{K}_{K}$ denote the share of the players playing $GrimK$ that have record less than $K$, let $\mu^{K}_{K'}$ be the share of $GrimK$ players with record less than $K'$, and let $\mu^{K'}_{K'}$ be the share of the players playing $GrimK'$ that have record less than $K'$. Then in an steady state we have

\begin{equation*}
\begin{split}
\mu^{K}_{K} & = 1 - \beta(\gamma, \varepsilon, (1 - \delta) \mu^{K}_{K} + \delta \mu^{K'}_{K})^{K} , \\
\mu^{K}_{K'} & = 1 - \beta(\gamma, \varepsilon, (1 - \delta) \mu^{K}_{K} + \delta \mu^{K'}_{K})^{K'} , \\
\mu^{K'}_{K} & = 1 - \gamma^{K - K'} \beta(\gamma, \varepsilon, (1 - \delta) \mu^{K}_{K'} + \delta \mu^{K'}_{K'})^{K'} , \\
\mu^{K'}_{K'} & = 1 - \beta(\gamma, \varepsilon, (1 - \delta) \mu^{K}_{K'} + \delta \mu^{K'}_{K'})^{K'}  .
\end{split}
\end{equation*}
This can be rewritten as
\begin{equation}
\begin{split}
f^{K}_{K}(\gamma,\varepsilon,\mu^{K}_{K},\mu^{K}_{K'},\mu^{K'}_{K},\mu^{K'}_{K'}) & := \mu^{K}_{K} + \beta(\gamma, \varepsilon, (1 - \delta) \mu^{K}_{K} + \delta \mu^{K'}_{K})^{K} - 1 = 0 , \\
f^{K}_{K'}(\gamma,\varepsilon,\mu^{K}_{K},\mu^{K}_{K'},\mu^{K'}_{K},\mu^{K'}_{K'}) & := \mu^{K}_{K'} + \beta(\gamma, \varepsilon, (1 - \delta) \mu^{K}_{K} + \delta \mu^{K'}_{K})^{K'} - 1 = 0 , \\
f^{K'}_{K}(\gamma,\varepsilon,\mu^{K}_{K},\mu^{K}_{K'},\mu^{K'}_{K},\mu^{K'}_{K'}) & := \mu^{K'}_{K} + \gamma^{K - K'} \beta(\gamma, \varepsilon, (1 - \delta) \mu^{K}_{K'} + \delta \mu^{K'}_{K'})^{K'} - 1 = 0 , \\
f^{K'}_{K'}(\gamma,\varepsilon,\mu^{K}_{K},\mu^{K}_{K'},\mu^{K'}_{K},\mu^{K'}_{K'}) & := \mu^{K'}_{K'} + \beta(\gamma, \varepsilon, (1 - \delta) \mu^{K}_{K'} + \delta \mu^{K'}_{K'})^{K'} - 1 = 0 .
\end{split}
\label{Steady-State Implicit Relation}
\end{equation}

Note that $\mu^{K}_{K} = \tilde{\mu}^{C}(\gamma,\varepsilon)$, $\mu^{K}_{K'} = 1 - \beta(\gamma,\varepsilon,\tilde{\mu}^{C}(\gamma,\varepsilon))^{K'}$, $\mu^{K'}_{K} = 1 - \gamma^{K - K'} \beta(\gamma,\varepsilon,1 - \beta(\gamma,\varepsilon,\tilde{\mu}^{C}(\gamma,\varepsilon))^{K'})^{K'}$, $\mu^{K'}_{K'} = 1 - \beta(\gamma,\varepsilon,1 - \beta(\gamma,\varepsilon,\tilde{\mu}^{C}(\gamma,\varepsilon))^{K'})^{K'}$ solves (\ref{Steady-State Implicit Relation}) when $\delta = 0$. The partial derivatives of the left-hand side of (\ref{Steady-State Implicit Relation}) evaluated at this point are given by
\begin{equation}
\begin{split}
& \begin{bmatrix}
\frac{\partial f^{K}_{K}}{\partial \mu^{K}_{K}} & \frac{\partial f^{K}_{K}}{\partial \mu^{K}_{K'}} & \frac{\partial f^{K}_{K}}{\partial \mu^{K'}_{K}} & \frac{\partial f^{K}_{K}}{\partial \mu^{K'}_{K'}} \\
\frac{\partial f^{K}_{K'}}{\partial \mu^{K}_{K}} & \frac{\partial f^{K}_{K'}}{\partial \mu^{K}_{K'}} & \frac{\partial f^{K}_{K'}}{\partial \mu^{K'}_{K}} & \frac{\partial f^{K}_{K'}}{\partial \mu^{K'}_{K'}} \\
\frac{\partial f^{K'}_{K}}{\partial \mu^{K}_{K}} & \frac{\partial f^{K'}_{K}}{\partial \mu^{K}_{K'}} & \frac{\partial f^{K'}_{K}}{\partial \mu^{K'}_{K}} & \frac{\partial f^{K'}_{K}}{\partial \mu^{K'}_{K'}} \\
\frac{\partial f^{K'}_{K'}}{\partial \mu^{K}_{K}} & \frac{\partial f^{K'}_{K'}}{\partial \mu^{K}_{K'}} & \frac{\partial f^{K'}_{K'}}{\partial \mu^{K'}_{K}} & \frac{\partial f^{K'}_{K'}}{\partial \mu^{K'}_{K'}} \\
\end{bmatrix} \\
= &
\begin{bmatrix}
1 + K \beta^{K - 1} \frac{\partial \beta}{\partial \mu^{C}} & 0 & 0 & 0 \\
K' \beta^{K' - 1} \frac{\partial \beta}{\partial \mu^{C}} & 1 & 0 & 0 \\
0 & \gamma^{K - K'} K' \beta^{K' - 1} \frac{\partial \beta}{\partial \mu^{C}} & 1 & 0 \\
0 & K' \beta^{K' - 1} \frac{\partial \beta}{\partial \mu^{C}} & 0 & 1
\end{bmatrix} .
\end{split}
\label{Partial Derivatives Matrix}
\end{equation}

Because $\tilde{\mu}^{C}(\gamma,\varepsilon) = 1 - \beta(\gamma,\varepsilon,\mu^{C}(\gamma,\varepsilon))^{K}$ and $K = \ln(1 - \tilde{\mu}^{C}(\gamma,\varepsilon)) / \ln(\beta(\gamma,\varepsilon,\tilde{\mu}^{C}(\gamma,\varepsilon)))$,
\begin{equation*}
\begin{split}
& 1 + K \beta(\gamma,\varepsilon,\tilde{\mu}^{C}(\gamma,\varepsilon))^{K - 1} \frac{\partial \beta}{\partial \mu^{C}}(\gamma,\varepsilon,\tilde{\mu}^{C}(\gamma,\varepsilon)) \\
= & 1 + \ln(1 - \tilde{\mu}^{C}(\gamma,\varepsilon))(1 - \tilde{\mu}^{C}(\gamma,\varepsilon)) \frac{\frac{\partial \beta}{\partial \mu^{C}}(\gamma,\varepsilon,\tilde{\mu}^{C}(\gamma,\varepsilon))}{\beta(\gamma,\varepsilon,\tilde{\mu}^{C}(\gamma,\varepsilon)) \ln(\beta(\gamma,\varepsilon,\tilde{\mu}^{C}(\gamma,\varepsilon)))} .
\end{split}
\end{equation*}

Recall that
\begin{equation*}
\beta(\gamma,\varepsilon,\mu^{C}) = \frac{\gamma (1 - (1 - \varepsilon) \mu^{C})}{1 - \gamma (1 - \varepsilon) \mu^{C}} = 1 - \frac{1 - \gamma}{1 - \gamma (1 - \varepsilon) \mu^{C}} .
\end{equation*}
Thus, $\lim_{(\gamma,\varepsilon) \rightarrow (1,0)} \beta(\gamma,\varepsilon,\tilde{\mu}^{C}(\gamma,\varepsilon)) = 1$. Hence, it follows that for high $\gamma$ and small $\varepsilon$, $\ln(\beta(\gamma,\varepsilon,\tilde{\mu}^{C}(\gamma,\varepsilon))) = - (1 - \beta(\gamma,\varepsilon,\tilde{\mu}^{C}(\gamma,\varepsilon))) + O(1 - \beta(\gamma,\varepsilon,\tilde{\mu}^{C}(\gamma,\varepsilon))^{2}$. Moreover,
\begin{equation*}
\begin{split}
\frac{\partial \beta}{\partial \mu^{C}}(\gamma,\varepsilon,\tilde{\mu}^{C}(\gamma,\varepsilon)) & = - \frac{(1 - \gamma) \gamma (1 - \varepsilon)}{(1 - \gamma (1 - \varepsilon) \mu^{C}(\gamma,\varepsilon))^{2}} \\
& = - \frac{\gamma (1 - \varepsilon)}{1 - \gamma (1 - \varepsilon) \tilde{\mu}^{C}(\gamma,\varepsilon)} (1 - \beta(\gamma,\varepsilon,\tilde{\mu}^{C}(\gamma,\varepsilon))) .
\end{split}
\end{equation*}
Combining these results gives us
\begin{equation*}
\lim_{(\gamma,\varepsilon) \rightarrow (1,0)} \frac{\frac{\partial \beta}{\partial \mu^{C}}(\gamma,\varepsilon,\tilde{\mu}^{C}(\gamma,\varepsilon))}{\beta(\gamma,\varepsilon,\tilde{\mu}^{C}(\gamma,\varepsilon)) \ln(\beta(\gamma,\varepsilon,\tilde{\mu}^{C}(\gamma,\varepsilon)))} = \frac{1}{1 - \tilde{\mu}^{C}} .
\end{equation*}
Since $\lim_{(\gamma,\varepsilon) \rightarrow (1,0)} \ln(1 - \tilde{\mu}^{C}(\gamma,\varepsilon))(1 - \tilde{\mu}^{C}(\gamma,\varepsilon)) = \ln(1 - \tilde{\mu}^{C}) (1 - \tilde{\mu}^{C})$, it further follows that
\begin{equation}
\lim_{(\gamma,\varepsilon) \rightarrow (1,0)} 1 + K \beta(\gamma,\varepsilon,\tilde{\mu}^{C}(\gamma,\varepsilon))^{K - 1} \frac{\partial \beta}{\partial \mu^{C}}(\gamma,\varepsilon,\tilde{\mu}^{C}(\gamma,\varepsilon))
= 1 + \ln(1 - \tilde{\mu}) .
\label{Limit Derivative Result}
\end{equation}

Since $\tilde{\mu} \neq 1 - 1/e$, we have $1 + \ln(1 - \tilde{\mu}) \neq 0$. Thus, using (\ref{Limit Derivative Result}), we conclude that the determinant of the matrix of partial derivatives in (\ref{Partial Derivatives Matrix}) is non-zero, and so can appeal to the implicit function theorem to conclude that for sufficiently high $\gamma$ and small $\varepsilon$, for each $K' \neq K$ and $\alpha > 0$, there is some $\delta_{1} > 0$ such that when the share of players playing $GrimK$ is $1 - \delta$ and the share of players playing $GrimK'$ is $\delta$ with $\delta < \delta_{1}$, there is a steady state where the fraction of players using $GrimK$ that are cooperators, $\mu^{C'}$, is such that $|\mu^{C'} - \tilde{\mu}^{C}(\gamma,\varepsilon)| < \alpha$. Additionally, because the $GrimK$ equilibrium with share of cooperators $\tilde{\mu}^{C}(\gamma,\varepsilon)$ is a strict equilibrium where players have uniformly strict incentives to play according to $GrimK$ at every own record and partner record, it follows that there is some $0 < \overline{\delta} < \delta_{1}$ such that, when the share of players playing $GrimK$ is $1 - \delta$ and the share of players playing $GrimK'$ is $\delta$ with $\delta < \overline{\delta}$, there is a steady state with share of cooperators $\mu^{C'}$ such that $|\mu^{C'} - \tilde{\mu}^{C}(\gamma,\varepsilon)| < \alpha$ where it is strictly optimal to play $GrimK$.

\end{proof}

\subsection{Dynamics} \label{section evolutionary dynamics}

We performed a simulation to capture dynamic evolution. We considered a population initially playing the $Grim5$ equilibrium with steady-state share of cooperators of $\mu^{C} \approx 0.8998$ when $\gamma = 0.9, \varepsilon = 0.1, g = 0.4, l = 2.8$ that is infected with a mutant population playing $Grim1$ at $t = 0$. The initial share of the population that played $Grim5$ was $.95$, and the complementary share of $0.05$ played $Grim1$. At $t = 0$, all of the $Grim1$ mutants had record $0$, while the record shares of the $Grim5$ population were proportional to those in the original steady state. At period $t$, the players match, observe each others' records (but not what population their opponent belongs to), and then play as their strategy dictates. We denote the average payoff of the $Grim5$ players and $Grim1$ players at period $t$ by $\pi^{Grim5,t}$ and $\pi^{Grim1,t}$, respectively. 

The evolution of the system from period $t - 1$ to $t$ was driven by the average payoffs and sizes of the two  populations at $t - 1$. In particular, at any period $t > 0$, the share of the newborn players that belonged to the $Grim5$ population ($\mu^{NGrim5,t}$) was proportional to the product of $\mu^{Grim5,t-1}$ and $\pi^{Grim5,t-1}$, and similarly the share of the $1 - \gamma$ newborn players that belonged to the $Grim1$ population ($\mu^{NGrim1,t}$) was proportional to the product of $\mu^{Grim1,t-1}$ and $\pi^{Grim1,t-1}$. Formally,
\begin{equation*}
\begin{split}
\mu^{NGrim5,t} & = \frac{\mu^{Grim5,t-1} \pi^{Grim5,t-1}}{\mu^{Grim5,t-1} \pi^{Grim5,t-1} + \mu^{Grim1,t-1} \pi^{Grim1,t-1}} (1 - \gamma) \\
\mu^{NGrim1,t} & = \frac{\mu^{Grim1,t-1} \pi^{Grim1,t-1}}{\mu^{Grim5,t-1} \pi^{Grim5,t-1} + \mu^{Grim1,t-1} \pi^{Grim1,t-1}} (1 - \gamma) .
\end{split}
\end{equation*}

\textbf{Extended Data Fig. \ref{Dynamics Figure}} presents the results of this simulation. \textbf{Extended Data Fig. \ref{Dynamics Figure}a} depicts the evolution of the share of players that use $Grim5$ and are cooperators (i.e. have record $k < 5$). Initially, this share is below the steady-state value of $\approx 0.8998$,  and is decreasing as the $Grim1$ mutants obtain high payoffs relative to the normal $Grim5$ players on average. However, the share of cooperator $Grim5$ players eventually begins to increase and approaches its steady-state value as the mutants die out.

The reason the mutants eventually die out is that their payoffs  eventually decline, as depicted in \textbf{Extended Data Fig. \ref{Dynamics Figure}b}. The tendency of the $Grim1$ players to defect means that they tend to move to high records relatively quickly, and so while they initially receive a high payoff from defecting against cooperators,  this advantage is short lived.

We found similar results when the mutant population plays $Grim9$ rather than $Grim1$, although the average payoff in the mutant population never exceeded that in the normal population. And we again found similar results when a population initially playing the $Grim8$ equilibrium with steady-state share of cooperators of $\mu^{C} \approx 0.613315$ and $\gamma = 0.95, \varepsilon = 0.05, g = 0.5, l = 4$ is infected with a mutant population playing $Grim3$ at $t = 0$, and for when it is infected with a mutant population playing $Grim13$.

\section{Public Goods}

Our analysis so far has taken the basic unit of social interaction to be the standard 2-player prisoner's dilemma. However, there are important social interactions that involve many players: the management of the commons and other public resources is a leading example [\citenum{hardin1968tragedy, dawes1980social, berkes1989benefits, ostrom1990governing}]. Such multiplayer public goods games have been the subject of extensive theoretical and experimental research [\citenum{milinski2002reputation, semmann2004strategic, ledyard1995handbook, fehr2000cooperation, olson2009logic, bergstrom1986private}]. Here we show that a simple variant of $GrimK$ strategies can support positive robust cooperation in the multiplayer public goods game when there is sufficient strategic complementarity.

We use the same model as considered so far, except that now in each period the players randomly match in groups of size $n$, for some fixed integer $n\geq 2$. All players in each group simultaneously decide whether to \emph{Contribute} ($C$) or \emph{Not Contribute} ($D$). If exactly $x$ of the $n$ players in the group contribute, each group member receives a benefit of $f(x)\geq0$, where $f:\mathbb{N} \to \mathbb{R}_{+}$ is a strictly increasing function with $f(0) = 0$. In addition, each player who contributes incurs a private cost of $c>0$. This coincides with the 2-player PD when $n = 2$, $f(1) = 1 + g$, $f(2) = l + 2 + g$, and $c = l + 1 + g$.

For each $x\in \{0,\ldots ,n-1\}$, let $\Delta(x)=f(x+1)-f(x)$ denote the marginal benefit to each member when there is an additional contribution. Assume that $\Delta(x)<c<n \Delta(x)$ for each $x\in \{0,\ldots ,n-1\}$. This assumption makes the public good game an $n$-player PD, in that $D$ is the selfishly optimal action while everyone playing $C$ is socially optimal.

We consider the same record system as in the 2-player PD: Newborns have record 0. If a player plays $D$, their record increases by 1. If a player plays $C$, their record increases by 1 with probability $\varepsilon >0$, and remains constant with probability $1- \varepsilon$.

As in the 2-player PD, we find that a key determinant of the prospects for robust cooperation is the degree of strategic complementarity or substitutability in the social dilemma. In the public good game, we say that the interaction exhibits \emph{strategic complementarity} if $\Delta(x)$ is increasing in $x$ (i.e., contributing is more valuable when more partners contribute), and exhibits \emph{strategic substitutability} if $\Delta(x)$ is decreasing in $x$.

We first show that with strategic substitutability the unique strict equilibrium is \emph{Never Contribute}. This generalizes our finding that \emph{Always Defect} is the unique strict equilibrium in the 2-player PD when $g \geq l$.

\begin{theorem}
For any $n\geq 2$, if the public good game exhibits strategic substitutability, the unique strict equilibrium is \emph{Never Contribute}.
\end{theorem}

\begin{proof}
Suppose $n$ players who all have the same record $k$ meet. By symmetry, either they all contribute or none of them contribute. In the former case, contributing is optimal for a record-$k$ player when all partners contribute, so by strategic substitutability contributing is also optimal for a record-$k$ player when a smaller number of partners contribute. Thus, a record-$k$ player contributes regardless of their partners' records. In the latter case, not contributing is optimal for a record-$k$ player when no partners contribute, so by strategic substitutability not contributing is also optimal for a record-$k$ player when a larger number of partners contribute. 

We have established that, for each $k$, record-$k$ players do not condition their behavior on their opponents' records. Hence, the distribution of future opposing actions faced by any player is independent of their record. This implies that not contributing is always optimal.
\end{proof}

We now turn to the case of strategic complementarity and consider the following simple generalization of $GrimK$ strategies: Records $k<K$ are considered to be ``good,'' while records $k\geq K$ are considered ``bad.'' When $n$ players meet, they all contribute if all of their records are good; otherwise, none of them contribute.

For $GrimK$ strategies to form an equilibrium, two incentive constraints must be satisfied: First, a player with record $0$ (the ``safest" good record) must want to contribute in a group with $n-1$ other good-record players. Second, a player with record $K-1$ (the ``most fragile" good record) must not want to contribute in a group where no one else contributes.

We let $g = c - \Delta(n - 1)$ and $l = c - \Delta(0)$. Note that
\begin{equation*}
V_{0} = (1 - \gamma) (\mu^{C})^{n - 1} (f(n) - c) + \gamma (1 - \varepsilon) (\mu^{C})^{n - 1} V_{0} + \gamma (1 - (1 - \varepsilon) (\mu^{C})^{n - 1}) V_{1} ,
\end{equation*}
which gives
\begin{equation*}
(1 - \varepsilon) \frac{\gamma}{1 - \gamma} (V_{0} - V_{1}) = \frac{1 - \varepsilon}{1 - (1 - \varepsilon) (\mu^{C})^{n - 1}} ((\mu^{C})^{n - 1} (f(n) - c) - V_{0}) .
\end{equation*}
By a similar argument to Lemma \ref{Value Functions Lemma}, it can be established that $V_{0} = \mu^{C} (\mu^{C})^{n - 1} (f(n) - c)$. We thus find that the cooperation constraint for a record $0$ player is
\begin{equation}
\frac{1 - \varepsilon}{1 - (1 - \varepsilon) (\mu^{C})^{n - 1}} (1 - \mu^{C}) (\mu^{C})^{n - 1} (f(n) - c) > g .
\label{C|C}
\end{equation}

In addition,
\begin{equation*}
V_{K - 1} = (1 - \gamma) (\mu^{C})^{n - 1} (f(n) - c) + \gamma (1 - \varepsilon) (\mu^{C})^{n - 1} V_{K - 1}
\end{equation*}
gives
\begin{equation*}
(1 - \varepsilon) \frac{\gamma}{1 - \gamma} V_{K - 1} = \frac{\gamma (1 - \varepsilon)}{1 - \gamma (1 - \varepsilon) (\mu^{C})^{n - 1}} (\mu^{C})^{n - 1} (f(n) - c) .
\end{equation*}
Thus, the defection constraint for a record $K - 1$ player is
\begin{equation*}
\frac{\gamma (1 - \varepsilon)}{1 - \gamma (1 - \varepsilon) (\mu^{C})^{n - 1}} (\mu^{C})^{n - 1} (f(n) - c) < l ,
\end{equation*}
which gives
\begin{equation}
(\mu^{C})^{n - 1} < \frac{1}{\gamma (1 - \varepsilon)} \frac{l}{f(n) - c + l} \Leftrightarrow \mu^{C} < \left( \frac{1}{\gamma (1 - \varepsilon)} \right)^{\frac{1}{n - 1}} \left( \frac{l}{f(n) - c + l} \right)^{\frac{1}{n - 1}} .
\label{Simplified (D|D)}
\end{equation}
This gives $\mu^{C} \leq (l / (f(n) - c + l))^{1/(n - 1)}$ in the $(\gamma, \varepsilon) \rightarrow (1,0)$ limit.

Moreover, in the limit where $\varepsilon \rightarrow 0$, (\ref{C|C}) gives
\begin{equation*}
\frac{1 - \mu^{C}}{1 - (\mu^{C})^{n - 1}} (\mu^{C})^{n - 1} (f(n) - c) \geq g \Leftrightarrow \frac{1}{\sum_{m = 0}^{n - 2} (\mu^{C})^{m}} (\mu^{C})^{n - 1} (f(n) - c) \geq g .
\end{equation*}
Note that $(\mu^{C})^{n - 1} / \sum_{m = 0}^{n - 2}(\mu^{C})^{m}$ is increasing in $\mu^{C}$. Thus, this inequality, along with the previous upper bound for $\mu^{C}$, puts the following requirement on the parameters:
\begin{equation*}
\frac{1 - ( \frac{l}{f(n) - c + l} )^{\frac{1}{n - 1}}}{\frac{f(n) - c}{f(n) - c + l}} \frac{l}{f(n) - c + l} (f(n) - c) \geq g ,
\end{equation*}
which simplifies to
\begin{equation}
g \leq \left(1 - \left( \frac{l}{f(n) - c + l} \right)^{\frac{1}{n - 1}} \right) l .
\label{Public Good g l Necessary Condition}
\end{equation}

So far we have established (\ref{Public Good g l Necessary Condition}), which is a necessary condition  on the $g,l$ parameters for any cooperation to be sustainable with $GrimK$ strategies in the $(\gamma,\varepsilon) \rightarrow (1,0)$ limit. We can further characterize the maximum limit share of cooperators in $GrimK$ equilibria using very similar arguments as those in Lemmas \ref{Equilibrium Characterization} and \ref{Denseness Result}.
\begin{theorem}
\begin{equation*}
\lim_{(\gamma,\varepsilon) \rightarrow (1,0)} \overline{\mu}^{C}_{n}(\gamma,\varepsilon) =
\begin{cases}
\left( \frac{l}{f(n) - c + l} \right)^{\frac{1}{n - 1}} & \text{if } g < \left(1 - \left( \frac{l}{f(n) - c + l} \right)^{\frac{1}{n - 1}} \right) l \\
0 & \text{if } g > \left(1 - \left( \frac{l}{f(n) - c + l} \right)^{\frac{1}{n - 1}} \right) l
\end{cases} .
\end{equation*}
\label{Public Goods Limit Performance Result}
\end{theorem}

Theorem \ref{Public Goods Limit Performance Result} shows that $GrimK$ strategies can support robust social cooperation in the $n$-player public goods game in much the same manner as in the 2-player PD. To see how this result reduces to Theorem \ref{Limit Performance Result} in the 2-player PD, note that $f(2) - c = 1$, so $(l / (f(n) - c + l))^{1/(n - 1)} = l / (1 + l)$ when $n=2$.

\renewcommand{\thesection}{A}

\renewcommand{\thesection}{R}

\end{document}